\documentclass[11pt,english]{article}

\usepackage{tikz}\usetikzlibrary{positioning,arrows,automata,shapes}
\usepackage{multirow}
\usepackage{mathtools}
\usepackage[colorlinks]{hyperref}
\usepackage{geometry}

\usepackage{fourier}
\usepackage{microtype}

\usepackage{amsmath,amsfonts,amssymb,amsthm}

\newtheorem{theorem}{Theorem}[section]
\newtheorem{definition}[theorem]{Definition}
\newtheorem{proposition}[theorem]{Proposition}
\theoremstyle{plain}

\newtheorem{lemma}[theorem]{Lemma}
\newtheorem{corollary}[theorem]{Corollary}

\DeclareMathOperator{\PTIME}{\mathbf{P}}
\DeclareMathOperator{\NP}{\mathbf{NP}}
\DeclareMathOperator{\co}{\mathbf{co-}\!}
\DeclareMathOperator{\EXPSPACE}{\mathbf{EXPSPACE}}
\DeclareMathOperator{\PSPACE}{\mathbf{PSPACE}}

\DeclareMathOperator{\Trk}{Trk}
\DeclareMathOperator{\Pref}{Pref}
\DeclareMathOperator{\Suff}{Suff}
\DeclareMathOperator{\states}{states}
\DeclareMathOperator{\lst}{lst}
\DeclareMathOperator{\fst}{fst}

\DeclareMathOperator{\hsA}{\langle A\rangle}
\DeclareMathOperator{\hsL}{\langle L\rangle}
\DeclareMathOperator{\hsB}{\langle B\rangle}
\DeclareMathOperator{\hsE}{\langle E\rangle}
\DeclareMathOperator{\hsD}{\langle D\rangle}
\DeclareMathOperator{\hsO}{\langle O\rangle}

\DeclareMathOperator{\hsG}{\langle G\rangle}
\DeclareMathOperator{\hsAt}{\langle \overline{A}\rangle}

\DeclareMathOperator{\hsBt}{\langle \overline{B}\rangle}
\DeclareMathOperator{\hsEt}{\langle \overline{E}\rangle}

\newcommand{\stat}{\mathsf{st}}
\newcommand{\LinearPast}{\mathsf{ct}}
\newcommand{\LinearTime}{\mathsf{lin}}

\newcommand{\A}{\mathsf{A}}
\newcommand{\Bbar}{\mathsf{\overline{B}}}
\newcommand{\B}{\mathsf{B}}
\newcommand{\E}{\mathsf{E}}

\newcommand{\Nat}{{\mathbb{N}}}

\DeclareMathAlphabet{\mathpzc}{OT1}{pzc}{m}{it}

\newcommand{\Length}{\textit{length}}
\newcommand{\Ku}{\ensuremath{\mathpzc{K}}}

\newcommand{\Lang}{{\textit{L}}}
\newcommand{\act}{{\textit{act}}}

\newcommand{\CTL}{\text{\sffamily CTL}}
\newcommand{\HS}{\text{\sffamily HS}}
\newcommand{\FO}{\text{\sffamily FO}}
\newcommand{\CTLStar}{\text{\sffamily CTL$^{*}$}}
\newcommand{\CTLStarLP}{\text{\sffamily CTL$^{*}_{lp}$}}
\newcommand{\LTL}{\text{\sffamily LTL}}
\newcommand{\LTLP}{\text{\sffamily LTL$_p$}}
\newcommand{\until}{\textsf{U}}
\newcommand{\Downarrowx}{\text{$\downarrow$$x$}}
\newcommand{\Downarrowy}{\text{$\downarrow$$y$}}
\newcommand{\Next}{\textsf{X}}
\newcommand{\Always}{\textsf{G}}
\newcommand{\Eventually}{\textsf{F}}
\newcommand{\EQ}{\exists}
\newcommand{\EQF}{\exists_f}
\newcommand{\AQ}{\forall}
\newcommand{\EQSubf}{\exists\textit{SubF}}
\newcommand{\present}{\textit{present}}

\begin{document}

\title{Interval vs.\ Point Temporal Logic Model Checking:\\ an Expressiveness Comparison%
\thanks{This work is an extended and revised version of \cite{fsttcs16}.}}
 
\author{Laura Bozzelli\textsuperscript{1}, Alberto Molinari\textsuperscript{2}, Angelo Montanari\textsuperscript{2},\\ Adriano Peron\textsuperscript{1}, Pietro Sala\textsuperscript{3}}
\date{\small \textsuperscript{1} University of Napoli \lq\lq Federico II\rq\rq, IT \quad \textsuperscript{2} University of Udine, IT \quad \textsuperscript{3} University of Verona, IT}

\maketitle

\begin{abstract}
In the last years, model checking with interval temporal logics is emerging as a viable alternative to  model checking with standard point-based temporal logics, such as $\LTL$, $\CTL$, $\CTLStar$, and the like. The behavior of the system is modelled by means of (finite) Kripke structures, as usual. However, while temporal logics which are interpreted  ``point-wise'' describe how the system evolves state-by-state, and predicate properties of system states, those which are interpreted ``interval-wise'' express properties of computation stretches, spanning a sequence of states. A proposition letter is assumed to hold over a computation stretch (interval) if and only if it holds over each component state (homogeneity assumption). 
%
%
A natural question arises: is there any advantage in replacing points by intervals as the primary temporal entities, or is it just a matter of taste? 

In this paper, we study the expressiveness of Halpern and Shoham's interval temporal logic ($\HS$) in model checking, in comparison with those of $\LTL$, $\CTL$, and $\CTLStar$. To this end, we consider three semantic variants of $\HS$: the state-based one, introduced by Montanari et al.\ in~\cite{DBLP:conf/time/MontanariMPP14,DBLP:journals/acta/MolinariMMPP16}, that allows time to branch both in the past and in the future, the computation-tree-based one, that allows time to branch in the future only, and the trace-based variant, that disallows time to branch.  These variants are compared
among themselves and to the aforementioned standard logics, getting a complete picture. In particular, we show that 
$\HS$ with trace-based semantics is equivalent to $\LTL$ (but at least exponentially more succinct), $\HS$ with  computation-tree-based semantics is equivalent to finitary $\CTLStar$, and $\HS$ with state-based semantics is incomparable with all of them ($\LTL$, $\CTL$, and $\CTLStar$).
\end{abstract}

\vfill

\noindent
The work has been supported by the GNCS project \emph{Formal Methods for Verification and Synthesis of Discrete and Hybrid Systems}. The work by A.\ Molinari and A.\ Montanari has also been supported by the project \emph{(PRID) ENCASE - Efforts in the uNderstanding of Complex interActing SystEms}.

\vfill

\newpage

\section{Introduction}

\emph{Point-based temporal logics} (PTLs) provide a 
standard framework for the specification of the 
behavior of reactive systems, that makes it possible to describe how a system evolves state-by-state (``point-wise'' 
view). PTLs have been successfully employed in \emph{model checking} (MC), which enables one to automatically 
verify complex finite-state systems
usually modelled 
as finite propositional Kripke structures. 
The MC methodology considers two types of PTLs---\emph{linear} and \emph{branching}---which differ 
in the underlying model of time.
In linear PTLs, like $\LTL$~\cite{pnueli1977temporal}, each moment in time has a unique possible future:
formulas are interpreted over paths of a Kripke structure, and thus they refer to a single computation of the system.
In branching PTLs,
like $\CTL$ and $\CTL^{*}$~\cite{emerson1986sometimes}, each moment in time may evolve into several possible futures:
 formulas are interpreted over states of the Kripke structure, hence referring to all the possible system computations.

\emph{Interval temporal logics} (ITLs) have been proposed as an alternative setting for reasoning about time~\cite{HS91,digital_circuits_thesis,Ven90}. Unlike standard PTLs, they assume intervals, instead of points, as their primitive entities. ITLs allow one to specify relevant temporal properties that involve, e.g., actions with duration, accomplishments, and temporal aggregations, which are inherently ``interval-based'', and thus cannot be naturally expressed by PTLs.
ITLs have been applied in various areas of computer science, including formal verification, computational linguistics, planning, and multi-agent systems
\cite{LM13,digital_circuits_thesis,DBLP:journals/ai/Pratt-Hartmann05}. 
\emph{Halpern and Shoham's modal logic of time intervals} (referred to as $\HS$)~\cite{HS91}
is the most popular among the ITLs.  It features one modality for each of the 13 possible ordering relations between pairs of intervals
(the so-called Allen's relations~\cite{All83}), apart from equality.
Its \emph{satisfiability problem} turns out to be highly undecidable for all interesting (classes of) linear orders~\cite{HS91}; the same happens with most of its fragments~%
\cite{DBLP:journals/amai/BresolinMGMS14,Lod00,MM14}, but there are some noteworthy exceptions like the logic of temporal neighbourhood $\mathsf{A\overline{A}}$, over all relevant (classes of) linear orders~\cite{BGMS09,DBLP:conf/tableaux/BresolinMSS11}, and the logic of sub-intervals $\mathsf{D}$, over the class of dense linear orders~\cite{BGMS10,DBLP:journals/corr/MontanariPS15}.

In this paper, we focus on the \emph{MC problem} for $\HS$. In order to check interval properties of computations, one needs to collect information about states into computation stretches, that is, finite paths of the Kripke structure (\emph{traces} for short). Each trace is interpreted as an interval, whose labelling is defined on the basis of the labelling of the component states. Such an approach to $\HS$ MC has been simultaneously and independently proposed by Montanari et al.\ in~\cite{DBLP:conf/time/MontanariMPP14,DBLP:journals/acta/MolinariMMPP16} and by Lomuscio and Michaliszyn in~\cite{LM13,LM14}. 

In \cite{DBLP:conf/time/MontanariMPP14,DBLP:journals/acta/MolinariMMPP16}, Montanari et al.\ assume a \emph{state-based} semantics, according to which
intervals/traces 
are \lq\lq forgetful\rq\rq{} of the history leading to their initial state. 
Since the initial (resp., final) state of an interval may feature several predecessors (resp., successors), such an interpretation induces a branching reference both in the future and in the past. A graphical account of the state-based semantics can be found in Figure \ref{fig:ST}; a detailed explanation will be given in the following.
The other fundamental choice done in \cite{DBLP:conf/time/MontanariMPP14,DBLP:journals/acta/MolinariMMPP16} concerns the labeling of intervals: a natural principle, known as the \emph{homogeneity assumption}, is adopted, which states that a proposition letter holds over an interval if and only if it holds over each component state (such an assumption turns out to be the most appropriate choice for many practical applications). 
In this setting, the MC problem for full $\HS$ turns out to be decidable. More precisely, it is  $\EXPSPACE$-hard~\cite{BMMPS16}, while the only known upper bound is non-elementary~\cite{DBLP:journals/acta/MolinariMMPP16}.%
\footnote{Here and in the following we refer to the \emph{combined complexity} of MC (which accounts for both the size of the Kripke structure and of the formula at the same time).} 
The exact complexity of MC for almost all the meaningful syntactic fragments of $\HS$, which ranges from $\co\NP$ to $\PTIME^{\NP}$, $\PSPACE$, and beyond, has been  determined in a subsequent series of papers~\cite{BMMPS16,BMMPS16B,DBLP:conf/icalp/BozzelliMMPS17,DBLP:journals/acta/MolinariMMPP16,MMP15B,MMP15,MMPS16}.

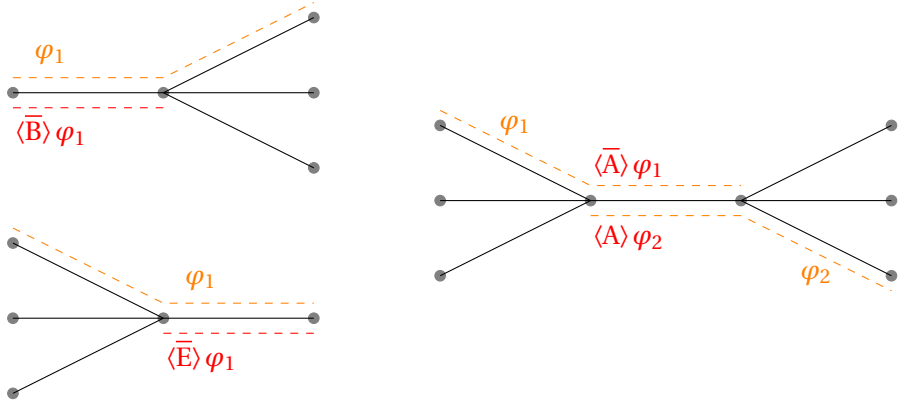
\begin{figure}[t]
\centering
\begin{minipage}{0.30\textwidth}
\begin{tikzpicture}
				\filldraw [gray] (0,2) circle (2pt)
				(2,2) circle (2pt)
				(4,3) circle (2pt)
				(4,2) circle (2pt)
				(4,1) circle (2pt);
					\filldraw [gray] (4,-1) circle (2pt)
				(2,-1) circle (2pt)
				(0,0) circle (2pt)
				(0,-1) circle (2pt)
				(0,-2) circle (2pt);
				\draw [black]  (0,2) -- (4,2);
				\draw [black]  (0,-1) -- (4,-1);
				\draw [black] (2,2) -- (4,3);
				\draw [black] (2,2) -- (4,1);
				\draw [black] (0,0) -- (2,-1);
				\draw [black] (0,-2) -- (2,-1);
				\draw [dashed, orange] (0,2.2) -> (2,2.2) -> (4,3.2);
				\draw [dashed, red] (0,1.8) -> (2,1.8);
				\draw [dashed, orange] (0,0.2) -> (2,-0.8) -> (4,-0.8);
				\draw [dashed, red] (2,-1.2) -> (4,-1.2);
				
				\node [orange] at (0.5,2.5) {$\varphi_1$};	
				\node [red] at (0.5,1.5) {$\hsBt\varphi_1$};
			\node [orange] at (2.5,-0.5) {$\varphi_1$};	
				\node [red] at (2.5,-1.5) {$\hsEt\varphi_1$};
						
				
			\end{tikzpicture}
\end{minipage}		
\hspace{1cm} 			
\begin{minipage}{0.45\textwidth}
			\begin{tikzpicture}
				%
					\filldraw [gray] (4,-1) circle (2pt)
				(2,-1) circle (2pt)
				(0,0) circle (2pt)
				(0,-1) circle (2pt)
				(0,-2) circle (2pt)
				(6,0) circle (2pt)
				(6,-1) circle (2pt)
				(6,-2) circle (2pt)	;
				\draw [black]  (0,-1) -- (6,-1);
				\draw [black] (4,-1) -- (6,0);
				\draw [black] (4,-1) -- (6,-2);
					\draw [black] (0,0) -- (2,-1);
				\draw [black] (0,-2) -- (2,-1);
				\draw [dashed, orange] (0,0.2) -> (2,-0.8) -> (4,-0.8);
				\draw [dashed, orange] (2,-1.2) -> (4,-1.2) -> (6,-2.2);
				
				\node [orange] at (1,0) {$\varphi_1$};	
				\node [red] at (2.5,-0.5) {$\hsAt\varphi_1$};
			\node [orange] at (5,-2) {$\varphi_2$};	
				\node [red] at (2.5,-1.5) {$\hsA \varphi_2$};
						
				
			\end{tikzpicture}
\end{minipage}
    \caption{State-based semantic variant $\HS_\stat$: past and future are branching.}
    \label{fig:ST}
\end{figure}


In~\cite{LM13,LM14}, Lomuscio and Michaliszyn address the MC problem for some fragments of $\HS$ extended with epistemic modalities. Their semantic assumptions are different from those made in~\cite{DBLP:conf/time/MontanariMPP14,DBLP:journals/acta/MolinariMMPP16}: the fragments are interpreted over the unwinding of the Kripke structure (\emph{computation-tree-based} semantics---see Figure~\ref{fig:CT} for a graphical account), and the interval labeling takes into account only the endpoints of intervals. 
In~\cite{LM13}, they focus on the $\HS$ fragment $\B\E$ of Allen's relations \emph{started-by} and \emph{finished-by}, extended with epistemic modalities. They consider a \emph{restricted form of MC (local MC)}, which checks the specification against a single (finite) initial computation interval, and prove that it is $\PSPACE$-complete.  
In~\cite{LM14}, they demonstrate that the picture drastically changes with other fragments of $\HS$ that allow one to access infinitely many intervals. In particular, they prove that the MC problem for the $\HS$ fragment $\A\Bbar$ of Allen's relations \emph{meets} and \emph{starts}, extended with epistemic modalities, is decidable with a non-elementary upper bound. The decidability status of MC for full epistemic $\HS$ 
is not known.

To summarize, the MC problem for \HS\ (and its fragments) has been extensively studied under the state-based and the computation-tree-based semantics, mainly focusing on complexity issues.  What is missing is a formal 
comparison of the expressiveness of \HS\ MC and MC for standard point-based temporal logics. A comparison of the expressiveness of the MC problem for \HS\ under the state-based and the computation-tree-based semantics
is missing as well.

\begin{figure}[t]
\centering
\begin{minipage}{0.3\textwidth}
\begin{tikzpicture}
				\filldraw [gray] (0,2) circle (2pt)
				(2,2) circle (2pt)
				(4,3) circle (2pt)
				(4,2) circle (2pt)
				(4,1) circle (2pt);
					\filldraw [gray] (4,-1) circle (2pt)
				(2,-1) circle (2pt)
				(0,-1) circle (2pt);
				%
				\draw [black]  (0,2) -- (4,2);
				\draw [black]  (0,-1) -- (4,-1);
				\draw [black] (2,2) -- (4,3);
				\draw [black] (2,2) -- (4,1);
				%
				\draw [dashed, orange] (0,2.2) -> (2,2.2) -> (4,3.2);
				\draw [dashed, red] (0,1.8) -> (2,1.8);
				\draw [dashed, orange] (0,-0.8) -> (2,-0.8) -> (4,-0.8);
				\draw [dashed, red] (2,-1.2) -> (4,-1.2);
				
				\node [orange] at (0.5,2.5) {$\varphi_1$};	
				\node [red] at (0.5,1.5) {$\hsBt\varphi_1$};
			\node [orange] at (2.5,-0.5) {$\varphi_1$};	
				\node [red] at (2.5,-1.5) {$\hsEt\varphi_1$};
						
				
			\end{tikzpicture}
\end{minipage}			
\hspace{1cm}
\begin{minipage}{0.45\textwidth}
			\begin{tikzpicture}
				%
					\filldraw [gray] (4,-1) circle (2pt)
				(2,-1) circle (2pt)
				(0,-1) circle (2pt)
				(6,0) circle (2pt)
				(6,-1) circle (2pt)
				(6,-2) circle (2pt)	;
				\draw [black]  (0,-1) -- (6,-1);
				\draw [black] (4,-1) -- (6,0);
				\draw [black] (4,-1) -- (6,-2);
				%
				\draw [dashed, orange] (0,-0.8) -> (2,-0.8) -> (4,-0.8);
				\draw [dashed, orange] (2,-1.2) -> (4,-1.2) -> (6,-2.2);
				
				\node [orange] at (1,-0.5) {$\varphi_1$};	
				\node [red] at (2.5,-0.5) {$\hsAt\varphi_1$};
			\node [orange] at (5,-2) {$\varphi_2$};	
				\node [red] at (2.5,-1.5) {$\hsA \varphi_2$};
						
				
			\end{tikzpicture}
\end{minipage}
    \caption{Computation-tree-based semantic variant $\HS_\LinearPast$: future is branching, past is linear, finite and cumulative.}
    \label{fig:CT}
\end{figure}
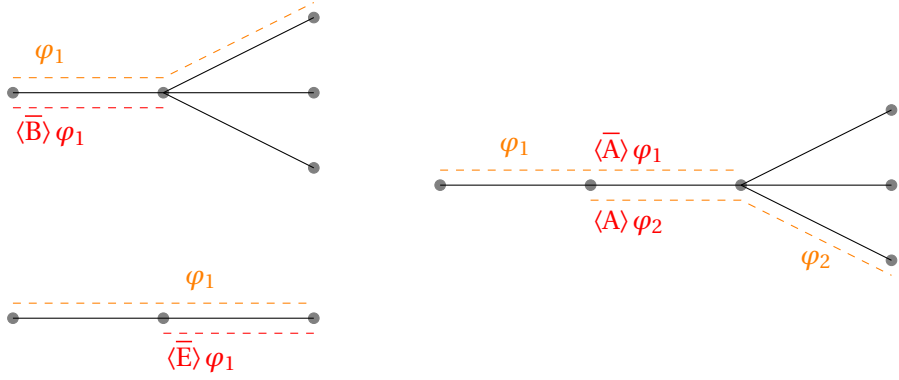

\paragraph{Our contribution.} 
In this paper, we study the expressiveness of $\HS$, in the context of MC, in comparison with that of the standard PTLs $\LTL$, $\CTL$, and $\CTLStar$. The analysis is carried on enforcing the homogeneity assumption.

We prove that $\HS$ endowed with the state-based semantics proposed in~\cite{DBLP:conf/time/MontanariMPP14,DBLP:journals/acta/MolinariMMPP16} (hereafter denoted as $\HS_\stat$) is not comparable with $\LTL$, $\CTL$, and $\CTLStar$. On the one hand, the result supports the intuition that $\HS_\stat$ gains some expressiveness by the ability of branching in the past. On the other hand, $\HS_\stat$ does not feature the possibility of forcing the verification of a property over an 
infinite path, thus implying that the formalisms are not comparable. With the aim of having a more ``effective'' comparison base, we consider two other semantic variants of $\HS$, 
namely, the \emph{computation-tree-based semantic variant} (denoted as $\HS_\LinearPast$) and the \emph{trace-based} one ($\HS_\LinearTime$). 

The state-based (see Figure~\ref{fig:ST}) and computation-tree-based (see Figure~\ref{fig:CT}) approaches rely on a branching-time setting and differ in the nature of past. In the latter approach, past is linear: each interval may have several possible futures, but only a unique past. Moreover, past is assumed to be finite 
and cumulative, that is, the story of the current situation increases with time, and is never forgotten. 
%
The trace-based approach relies on a linear-time setting  (see Figure~\ref{fig:LN}), where the infinite paths (computations) of the given Kripke structure are the main semantic entities. Branching is neither allowed in the past nor in the future.
Note that the linear-past (rather than branching) approach is more suited to the specification  of dynamic behaviors, because it considers states in a computation tree, while the branching-past approach considers machine states, where past is not very meaningful for the specification of behavioral constraints~\cite{LS95}.

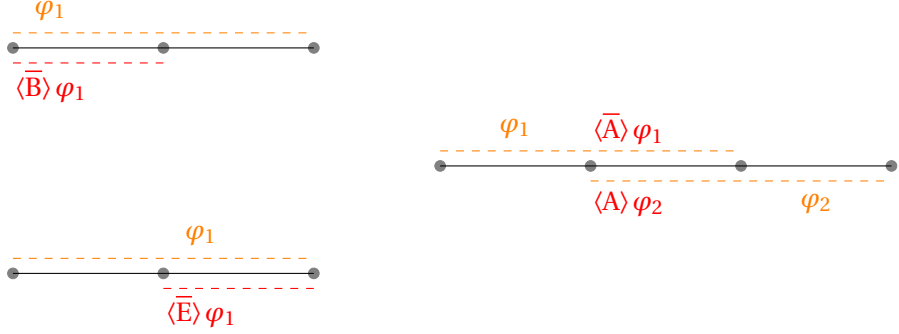
\begin{figure}[t]
\centering
\begin{minipage}{0.3\textwidth}
\begin{tikzpicture}
				\filldraw [gray] (0,2) circle (2pt)
				(2,2) circle (2pt)
				(4,2) circle (2pt);
				%
					\filldraw [gray] (4,-1) circle (2pt)
				(2,-1) circle (2pt)
				(0,-1) circle (2pt);
				%
				\draw [black]  (0,2) -- (4,2);
				\draw [black]  (0,-1) -- (4,-1);
				%
				%
				\draw [dashed, orange] (0,2.2) -> (2,2.2) -> (4,2.2);
				\draw [dashed, red] (0,1.8) -> (2,1.8);
				\draw [dashed, orange] (0,-0.8) -> (2,-0.8) -> (4,-0.8);
				\draw [dashed, red] (2,-1.2) -> (4,-1.2);
				
				\node [orange] at (0.5,2.5) {$\varphi_1$};	
				\node [red] at (0.5,1.5) {$\hsBt\varphi_1$};
			\node [orange] at (2.5,-0.5) {$\varphi_1$};	
				\node [red] at (2.5,-1.5) {$\hsEt\varphi_1$};
						
				
			\end{tikzpicture}
\end{minipage}
\hspace{1cm}
\begin{minipage}{0.45\textwidth}
        \begin{tikzpicture}
				%
					\filldraw [gray] (4,-1) circle (2pt)
				(2,-1) circle (2pt)
				(0,-1) circle (2pt)
				(6,-1) circle (2pt);
				%
				\draw [black]  (0,-1) -- (6,-1);
				%
				%
				\draw [dashed, orange] (0,-0.8) -> (2,-0.8) -> (4,-0.8);
				\draw [dashed, orange] (2,-1.2) -> (4,-1.2) -> (6,-1.2);
				
				\node [orange] at (1,-0.5) {$\varphi_1$};	
				\node [red] at (2.5,-0.5) {$\hsAt\varphi_1$};
			\node [orange] at (5,-1.5) {$\varphi_2$};	
				\node [red] at (2.5,-1.5) {$\hsA \varphi_2$};
						
				
			\end{tikzpicture}
\end{minipage}
    \caption{Trace-based semantic variant $\HS_\LinearTime$: neither past nor future are branching.}
    \label{fig:LN}
\end{figure}

The variant $\HS_\LinearPast$ is a natural candidate for an expressiveness comparison with the branching time logics  $\CTL$ and $\CTLStar$. The most interesting and technically involved result is the characterization of the expressive power of $\HS_\LinearPast$: $\HS_\LinearPast$ turns out to be expressively equivalent to finitary $\CTLStar$, that is, the variant of $\CTLStar$ with quantification over finite paths. As for $\CTL$, a non comparability result can be stated.

The variant $\HS_\LinearTime$ is a natural candidate for an expressiveness comparison with $\LTL$. 
We prove that $\HS_\LinearTime$ and $\LTL$ are equivalent (this result holds true even for a very small fragment of $\HS_\LinearTime$), but the former is at least exponentially more succinct than the latter. 

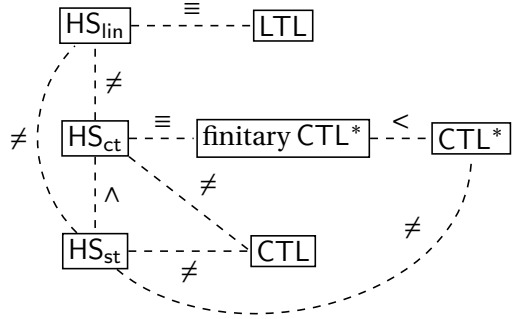
\begin{figure}[t]
\centering
\begin{tikzpicture}[-,>=stealth',shorten >=1pt,auto,semithick,main node/.style={rectangle,draw,inner sep=2pt}]  
\tikzstyle{gray node}=[fill=gray!30]
\node [main node](0) at (0,0) {$\HS_\LinearPast$};
\node [main node](1) at (0,1.5) {$\HS_\LinearTime$};
\node [main node](2) at (0,-1.5) {$\HS_\stat$};
\node [main node](3) at (2.5,0) {finitary $\CTLStar$};
\node [main node](4) at (2.5,1.5) {$\LTL$};
\node [main node](6) at (2.5,-1.5) {$\CTL$};
\node [main node](5) at (5,0) {$\CTLStar$};
\draw [dashed] (0.east) to node {$\equiv$} (3);
\draw [dashed] (1.east) to node {$\equiv$} (4);
\draw [dashed] (3.east) to node {$<$} (5);
\draw [dashed] (0.north) to node [swap] {$\neq$} (1);
\draw [dashed] (0.south) to node {\rotatebox{-90}{$<$}} (2);
\draw [dashed] (2.east) to node [swap] {$\neq$} (6.west);
\draw [dashed] (2) [out=-40,in=270] to node [near end] {$\neq$} (5);
\draw [dashed] (0.south east) to node {$\neq$} (6.west);
\draw [dashed] (2) [out=135,in=225] to node {$\neq$} (1);
\end{tikzpicture}
\vspace*{-0.6cm}
\caption{Overview of the expressiveness results.}\label{results}
\end{figure}

We complete the picture with a comparison of the three semantic variants $\HS_\stat$, $\HS_\LinearPast$, and $\HS_\LinearTime$. We prove that, as expected, $\HS_\LinearTime$ is not comparable with either of the branching versions, $\HS_\LinearPast$ and $\HS_\stat$. The interesting result is that, on the other hand, $\HS_\LinearPast$ is strictly included in $\HS_\stat$: this supports $\HS_\stat$, adopted in~\cite{DBLP:journals/acta/MolinariMMPP16,MMP15B,MMP15,MMPS16,BMMPS16,BMMPS16B}, as a reasonable and adequate semantic choice.
The complete picture of the expressiveness results is reported in Figure~\ref{results}
(the symbols $\neq$, $\equiv$, and $<$ denote incomparability, equivalence, and strict 
inclusion, respectively).

\paragraph{Structure of the paper.}
%
In Section~\ref{sec:backgr}, we introduce basic notation and
preliminary notions. In Subsection~\ref{sect:Kripke} we define Kripke structures and interval structures, 
in Subsection~\ref{sect:PTL} we recall the well-known PTLs $\LTL$, $\CTL$, and $\CTLStar$, and in Subsection \ref{sect:HS} we present the interval temporal logic $\HS$. Then, in Subsection~\ref{sect:3sem} we define the three semantic variants of $\HS$ ($\HS_\stat$, $\HS_\LinearPast$, and $\HS_\LinearTime$). Finally, in Subsection~\ref{subs:vendingMach} we provide a detailed example which gives an intuitive account of the three semantic variants and highlights their differences. 
In the next three sections, we analyze and compare their expressiveness.
In Section~\ref{sec:CharacterizezionOfLeniarTimeHS} we show the expressive equivalence of $\LTL$ and $\HS_\LinearTime$. Then, in Section~\ref{sec:characterizationHSLinearPast} we prove the expressive equivalence of $\HS_\LinearPast$ and finitary $\CTLStar$. Finally, in Section~\ref{sect:allSems} we compare the expressiveness of 
$\HS_\stat$, $\HS_\LinearPast$, and $\HS_\LinearTime$.
Conclusions summarize the work done and outline some directions for future research.
\section{Preliminaries}\label{sec:backgr}

In this section, we introduce the notation and some fundamental notions that will be extensively used in the rest of the paper.
Let $(\Nat,<)$ be the set of natural numbers equipped with the standard linear ordering. For all  $i,j\in\Nat $, with $i\leq j$, we denote by $[i,j]$ the set of natural numbers $h$ such that $i\leq h\leq j$.
Let $\Sigma$ be an alphabet and $w$ be a non-empty finite or infinite word over $\Sigma$. We denote by $|w|$ the length of $w$ ($|w|=\infty$ if $w$ is infinite). For all  $i,j\in\Nat $, with $i\leq j$, $w(i)$ denotes the
$i$-th letter of $w$, while $w[i,j]$ denotes the finite subword of $w$ given by $w(i)\cdots w(j)$. If $w$ is finite and $|w|=n+1$, we define
$\fst(w)=w(0)$ and $\lst(w)=w(n)$.
The sets of all proper prefixes and suffixes of $w$ are
$\Pref(w)=\{w[0,i] \mid 0\leq i\leq n-1\}$ and $\Suff(w)=\{w[i,n] \mid 1\leq i\leq n\}$, respectively.
The set of all the finite words over $\Sigma$ is denoted by $\Sigma^*$, and $\Sigma^+:=\Sigma^*\setminus\{\varepsilon\}$, where $\varepsilon$ is the empty word.

\subsection{Kripke structures and interval structures}\label{sect:Kripke}
Systems are usually modelled as Kripke structures. Let $\mathpzc{AP}$ be a finite set of proposition letters, which represent predicates decorating the states of the given system.


\begin{definition}[Kripke structure]
A \emph{Kripke structure} over  a finite set $\mathpzc{AP}$ of proposition letters is a tuple $\mathpzc{K}=(\mathpzc{AP},S, \delta,\mu,s_0)$, where  $S$ is a set of states,
$\delta\subseteq S\times S$ is a left-total transition relation, 
$\mu:S\to 2^\mathpzc{AP}$ is a total labelling function assigning to each state $s$ the set of proposition letters that hold over it, and $s_0\in S$ is the initial state. For $(s,s')\in \delta$, we say that $s'$ is a successor of $s$,
and $s$ is a predecessor of $s'$. Finally, we say that $\mathpzc{K}$ is finite if $S$ is finite.
\end{definition}


\begin{figure}[b]
\centering
\begin{tikzpicture}[->,>=stealth,thick,shorten >=1pt,auto,node distance=2cm,every node/.style={circle,draw}]
    \node [style={double}](v0) {$\stackrel{s_0}{p}$};
    \node (v1) [right of=v0] {$\stackrel{s_1}{q}$};
    \draw (v0) to [bend right] (v1);
    \draw (v1) to [bend right] (v0);
    \draw (v1) to [loop right] (v1);
\end{tikzpicture}
\caption{The Kripke structure $\mathpzc{K}$.}\label{KEquiv}
\end{figure}
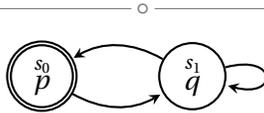
For example, Figure~\ref{KEquiv} depicts the finite Kripke structure $\mathpzc{K}=(\{p,q\},\{s_0,s_1\},\delta,\mu,s_0)$,
where $\delta\!=\!\{(s_0,s_1),(s_1,s_0),(s_1,s_1)\}$, 
$\mu(s_0)\!=\!\{p\}$, $\mu(s_1)\!=\!\{q\}$.
The initial state $s_0$ is marked by a double circle. 

Let $\mathpzc{K}=(\mathpzc{AP},S, \delta,\mu,s_0)$ be a Kripke structure. An infinite path $\pi$ of $\mathpzc{K}$ is an infinite word over $S$ such that $(\pi(i),\pi(i+1))\in \delta$ for all $i\geq 0$.
A \emph{trace} (or finite path) of $\mathpzc{K}$ is a non-empty prefix of some infinite path of $\mathpzc{K}$. A finite or infinite path is \emph{initial} if it starts from the initial state of $\mathpzc{K}$.
Let $\Trk_\mathpzc{K}$ be the (\emph{infinite}) set of all traces of $\mathpzc{K}$ and
$\Trk_\mathpzc{K}^{0}$ be the set of initial traces of $\mathpzc{K}$.
For a trace $\rho$, $\states(\rho)$ denotes the set of states occurring in $\rho$, i.e., $\states(\rho)\!=\!\{\rho(0),\ldots,\rho(n)\}$, where $|\rho|\!=\! n+1$.

%

We now introduce the notion of \emph{$D$-tree structure}, namely, an infinite tree-shaped Kripke structure with branches over a set $D$ of directions.

\begin{definition}[$D$-tree structure] Given a set $D$ of directions,
a \emph{$D$-tree structure} (over $\mathpzc{AP}$) is a Kripke structure $\mathpzc{K}=(\mathpzc{AP},S, \delta,\mu,s_0)$
such that $s_0\in D$, $S$ is a prefix closed subset of $D^{+}$, and $\delta$ is the set of pairs $(s,s')\in S\times S$ such that there exists $d\in D$
for which $s'=s\cdot d$ (note that $\delta$ is completely specified by $S$). The states of a $D$-tree structure are called \emph{nodes}.
\end{definition}

A Kripke structure $\mathpzc{K}=(\mathpzc{AP},S, \delta,\mu,s_0)$ induces an $S$-tree structure, called the \emph{computation tree of $\mathpzc{K}$}, denoted by
$\mathpzc{C}(\mathpzc{K})$, which is obtained by unwinding $\mathpzc{K}$ 
from the initial state (note that the directions are the set of states of $\mathpzc{K}$).
Formally, $\mathpzc{C}(\mathpzc{K})= (\mathpzc{AP},\Trk_\mathpzc{K}^{0}, \delta',\mu',s_0)$, where the set of nodes is the set of initial traces of
$\mathpzc{K}$ and  for all $\rho,\rho'\in \Trk_\mathpzc{K}^{0}$, $\mu'(\rho)=\mu(\lst(\rho))$ and $(\rho,\rho')\in \delta'$ if and only if
$\rho'=\rho\cdot s$ for some $s\in S$. See Figure~\ref{fig:unr} for an example.
\begin{figure}[t]
    \centering
\begin{tikzpicture}[->,>=stealth,thick,shorten >=1pt,auto,node distance=1.6cm,every node/.style={ellipse,draw}]
			\node[double] (v0) {$s_0$};
			\node (v1) at (0,-1) {$s_0s_1$};
			\node (v01) at (-1,-2) {$s_0s_1s_0$};
			\node (v11) at (1,-2) {$s_0s_1s_1$};
			\node (v02) at (-1,-3) {$s_0s_1s_0s_1$};
			\node (v022) at (1,-3) {$s_0s_1s_1s_0$};
			\node (v12) at (3,-3) {$s_0s_1s_1s_1$};
			

			\node (v30) [draw=none, below =0.2cm of v02] {\Large $\cdots$};
			\node (v31) [draw=none, below left =0.21cm and 0.09cm of v12] {\Large $\cdots$};

			\draw (v0) to (v1);
			\draw (v1) to (v01); \draw (v1) to (v11);
			\draw (v01) to (v02);
			\draw (v11) to (v022);
			\draw (v11) to (v12);
		\end{tikzpicture}
	\vspace{-0.3cm}
    \caption{Computation tree $\mathpzc{C}(\mathpzc{K})$ of the Kripke structure $\Ku$ of Figure~\ref{KEquiv}.}
    \label{fig:unr}
\end{figure}
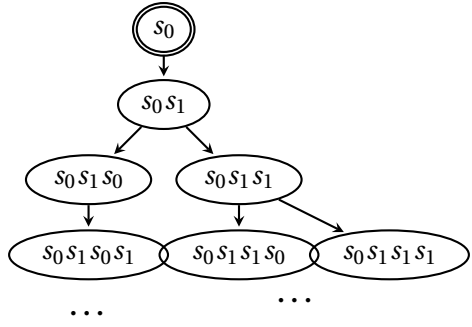

Given a strict partial ordering $\mathbb{S}=(X,<)$, an \emph{interval} in $\mathbb{S}$ is an ordered pair
$[x,y]$ such that $x,y\in X$ and $x\leq y$. The interval $[x,y]$ denotes the subset of $X$ given by the set of points $z\in X$ such
that $x\leq z\leq y$. We denote by $\mathbb{I}(\mathbb{S})$ the set of intervals in $\mathbb{S}$.

\begin{definition}[Interval structure]
An \emph{interval structure} $\mathpzc{IS}$ over  $\mathpzc{AP}$ is a pair $\mathpzc{IS}=(\mathbb{S},\sigma)$ such that $\mathbb{S}=(X,<)$ is a strict partial ordering
and $\sigma: \mathbb{I}(\mathbb{S}) \to 2^\mathpzc{AP}$ is a labeling function assigning a set of proposition letters
to each interval over $\mathbb{S}$.
\end{definition}

\subsection{Standard temporal logics}\label{sect:PTL}

In this subsection, we recall the standard propositional temporal logics $\CTLStar$, $\CTL$,  and $\LTL$~\cite{emerson1986sometimes,pnueli1977temporal}.
Given a set of proposition letters $\mathpzc{AP}$, the formulas $\varphi$ of
$\CTLStar$ are defined as follows:
\[
\varphi ::= \top \ | \ p \ | \ \neg \varphi \ | \ \varphi \wedge \varphi \ | \ \Next \varphi\ | \ \varphi \until \varphi\ | \ \EQ  \varphi ,
\]
where $p\in \mathpzc{AP}$, $\Next$ and $\until$ are the
``next'' and ``until'' temporal modalities,   and $\EQ$ is the
 existential path quantifier.
\footnote{Hereafter, we denote by $\exists/\forall$ the existential/universal path quantifiers (instead of by the usual E/A), in order not to confuse them with the $\HS$ modalities $\mathsf{E}/\mathsf{A}$.}
 We also use the standard shorthands $\AQ\varphi:=\neg\EQ\neg\varphi$ (``universal path quantifier''),
$\Eventually\varphi:= \top \until \varphi$ (``eventually'' or ``in the future'') and its
dual $\Always \varphi:=\neg \Eventually\neg\varphi$ (``always'' or ``globally'').
Hereafter, we denote by $|\varphi|$ the size of $\varphi$, that is, the number of its symbols/subformulas.

The logic $\CTL$ is the  fragment of $\CTLStar$ where each temporal modality is immediately preceded by a path quantifier, whereas  $\LTL$ corresponds to the path-quantifier-free  fragment of $\CTLStar$.

Given a Kripke
structure $\mathpzc{K}=(\mathpzc{AP},S,\delta,\mu,s_0)$, an infinite path $\pi$ of
$\mathpzc{K}$, and a position $i\geq 0$ along $\pi$, the
satisfaction relation $\mathpzc{K},\pi,i \models \varphi$ for
$\CTLStar$, written simply $\pi,i \models \varphi$ when $\mathpzc{K}$ is clear from the context, is defined as follows (Boolean connectives are treated as usual):
\[ \begin{array}{ll}
\pi, i \models p  &  \Leftrightarrow  p \in \mu(\pi(i)),\\
\pi, i \models \Next \varphi  & \Leftrightarrow   \pi, i+1 \models \varphi ,\\
\pi, i \models \varphi_1\until \varphi_2  &
  \Leftrightarrow  \text{for some $j\geq i$}: \pi, j
  \models \varphi_2
  \text{ and }  \pi, k \models  \varphi_1 \text{ for all }i\leq k<j,\\
\pi, i \models \EQ \varphi  & \Leftrightarrow \text{for some infinite path } \pi'  \text{ starting from $\pi(i)$, }  \pi', 0 \models \varphi .
\end{array} \]
The model checking (MC) problem is defined as follows: $\mathpzc{K}$ is a model of $\varphi$, written $\mathpzc{K}\models \varphi$, if for all initial infinite paths $\pi$ of $\mathpzc{K}$, it holds that $\mathpzc{K},\pi, 0 \models\varphi$.

We also consider a variant of $\CTLStar$, called \emph{finitary} $\CTLStar$, where the path quantifier $\EQ$ of $\CTLStar$ is replaced by the finitary path
quantifier $\EQF$. In this setting, path quantification ranges over the traces (finite paths) starting from the current state.
The satisfaction relation $\rho, i \models \varphi$, where $\rho$ is a trace and $i$ is a position along $\rho$, is similar to that given for $\CTLStar$ with the only difference of finiteness of paths, and the fact that for a formula $\Next\varphi$, 
  $\rho, i \models \Next\varphi$ if and only if $i+1<|\rho|$ and $\rho, i+1 \models \varphi$.    A Kripke structure $\mathpzc{K}$  is
  a model of a finitary $\CTLStar$ formula if for each initial trace $\rho$ of  $\mathpzc{K}$,  it holds that $\mathpzc{K},\rho, 0 \models\varphi$.
  
The MC problem for both $\CTLStar$ and $\LTL$ is $\PSPACE$-complete~\cite{DBLP:conf/popl/EmersonL85,DBLP:journals/jacm/SistlaC85}. It is not difficult to show that, as it happens with finitary $\LTL$~\cite{DBLP:conf/ijcai/GiacomoV13}, MC for finitary $\CTLStar$ is $\PSPACE$-complete as well.

\subsection{The interval temporal logic $\HS$}\label{sect:HS}
An interval algebra was proposed by Allen in~\cite{All83} to reason about intervals and their relative order, while a systematic logical study of interval representation and reasoning was done a few years later by Halpern and Shoham, that introduced the interval temporal logic $\HS$ featuring one modality for each Allen
relation, but equality~\cite{HS91}.
Table~\ref{allen} depicts 6 of the 13 Allen's relations,
together with the corresponding $\HS$ (existential) modalities. The other 7 relations are the 6 inverse relations (given a binary relation $\mathpzc{R}$, the inverse relation $\overline{\mathpzc{R}}$ is such that $b \overline{\mathpzc{R}} a$ if and only if $a \mathpzc{R} b$) and equality.

\begin{table}[b]
\centering
\caption{Allen's relations and corresponding $\HS$ modalities.}\label{allen}
\resizebox{\width}{\height}{
\begin{tabular}{cclc}
\hline
\rule[-1ex]{0pt}{3.5ex} Allen relation & $\HS$ & Definition w.r.t. interval structures &  Example\\
\hline

&   &   & \multirow{7}{*}{\begin{tikzpicture}[scale=0.96]
\draw[draw=none,use as bounding box](-0.3,0.2) rectangle (3.3,-3.1);
\coordinate [label=left:\textcolor{red}{$x$}] (A0) at (0,0);
\coordinate [label=right:\textcolor{red}{$y$}] (B0) at (1.5,0);
\draw[red] (A0) -- (B0);
\fill [red] (A0) circle (2pt);
\fill [red] (B0) circle (2pt);
\coordinate [label=left:$v$] (A) at (1.5,-0.5);
\coordinate [label=right:$z$] (B) at (2.5,-0.5);
\draw[black] (A) -- (B);
\fill [black] (A) circle (2pt);
\fill [black] (B) circle (2pt);
\coordinate [label=left:$v$] (A) at (2,-1);
\coordinate [label=right:$z$] (B) at (3,-1);
\draw[black] (A) -- (B);
\fill [black] (A) circle (2pt);
\fill [black] (B) circle (2pt);
\coordinate [label=left:$v$] (A) at (0,-1.5);
\coordinate [label=right:$z$] (B) at (1,-1.5);
\draw[black] (A) -- (B);
\fill [black] (A) circle (2pt);
\fill [black] (B) circle (2pt);
\coordinate [label=left:$v$] (A) at (0.5,-2);
\coordinate [label=right:$z$] (B) at (1.5,-2);
\draw[black] (A) -- (B);
\fill [black] (A) circle (2pt);
\fill [black] (B) circle (2pt);
\coordinate [label=left:$v$] (A) at (0.5,-2.5);
\coordinate [label=right:$z$] (B) at (1,-2.5);
\draw[black] (A) -- (B);
\fill [black] (A) circle (2pt);
\fill [black] (B) circle (2pt);
\coordinate [label=left:$v$] (A) at (1.3,-3);
\coordinate [label=right:$z$] (B) at (2.3,-3);
\draw[black] (A) -- (B);
\fill [black] (A) circle (2pt);
\fill [black] (B) circle (2pt);
\coordinate (A1) at (0,-3);
\coordinate (B1) at (1.5,-3);
\draw[dotted] (A0) -- (A1);
\draw[dotted] (B0) -- (B1);
\end{tikzpicture}}\\

\textsc{meets} & $\hsA$ & $[x,y]\mathpzc{R}_A[v,z]\iff y=v$ &\\

\textsc{before} & $\hsL$ & $[x,y]\mathpzc{R}_L[v,z]\iff y<v$ &\\

\textsc{started-by} & $\hsB$ & $[x,y]\mathpzc{R}_B[v,z]\iff x=v\wedge z<y$ &\\

\textsc{finished-by} & $\hsE$ & $[x,y]\mathpzc{R}_E[v,z]\iff y=z\wedge x<v$ &\\

\textsc{contains} & $\hsD$ & $[x,y]\mathpzc{R}_D[v,z]\iff x<v\wedge z<y$ &\\

\textsc{overlaps} & $\hsO$ & $[x,y]\mathpzc{R}_O[v,z]\iff x<v<y<z$ &\\

\hline
\end{tabular}}
\end{table}

For a set of proposition letters $\mathpzc{AP}$, the formulas $\psi$ of
$\HS$ are defined as follows:
\[
    \psi ::= p \;\vert\; \neg\psi \;\vert\; \psi \wedge \psi \;\vert\; \langle X\rangle\psi, 
\]
where $p\in\mathpzc{AP}$ and $X\in\{A,L,B,E,D,O,\overline{A},\overline{L},\overline{B},\overline{E},\overline{D},
\overline{O}\}$. For any modality $\langle X\rangle$, the dual universal modality $[X]\psi$ is defined as $\neg\langle X\rangle\neg\psi$.
For any subset of Allen's relations $\{X_1,\ldots,X_n\}$,  $\mathsf{X_1 \cdots X_n}$  denotes the $\HS$ fragment featuring (universal and existential) modalities for $X_1,\ldots, X_n$ only.


We assume the \emph{non-strict semantic version of $\HS$}, which admits intervals consisting of a single point.\footnote{All the results we prove in the paper hold for the strict version as well.} Under such an assumption, all $\HS$ modalities can be expressed in terms of 
$\hsB, \hsE, \hsBt$, and $\hsEt$ \cite{Ven90}. As an example, $\hsA$ can be expressed in terms of $\hsE$ and $\hsBt$ as: $\hsA \varphi:= ([E]\bot \wedge (\varphi \vee \hsBt \varphi)) \vee \hsE ([E]\bot \wedge (\varphi \vee \hsBt \varphi))$. We also use the derived operator $\hsG$ of $\HS$ (and its dual $[G]$), which allows one to select  arbitrary subintervals of a given interval, and is defined as: $\hsG\psi:= \psi \vee \hsB\psi \vee \hsE\psi \vee \hsB\hsE\psi$.

$\HS$ can be viewed as a multi-modal logic with $\hsB, \hsE, \hsBt$, and $\hsEt$ as primitive modalities
and its semantics can be defined over a multi-modal Kripke structure, called \emph{abstract interval model}, where intervals are treated as atomic objects and Allen's relations as binary relations over intervals.

\begin{definition}[Abstract interval model~\cite{DBLP:journals/acta/MolinariMMPP16}]
An \emph{abstract interval model} over $\mathpzc{AP}$ is a tuple $\mathpzc{A}=(\mathpzc{AP},\mathbb{I},B_\mathbb{I},E_\mathbb{I},\allowbreak \sigma)$, where
     $\mathbb{I}$ is a set of worlds,
     $B_\mathbb{I}$  and $E_\mathbb{I}$ are two binary relations over $\mathbb{I}$, and
     $\sigma:\mathbb{I}\to 2^{\mathpzc{AP}}$ is a labeling function assigning a set of proposition letters to each world.
\end{definition}


Let  $\mathpzc{A}=(\mathpzc{AP},\mathbb{I},   B_\mathbb{I},E_\mathbb{I},\sigma)$
be an abstract interval model. In the interval setting, $\mathbb{I}$ is interpreted as a set of intervals, $B_\mathbb{I}$ and $E_\mathbb{I}$ as Allen's relations $B$ (\emph{started-by})  and $E$ (\emph{finished-by}), respectively, and $\sigma$ assigns to each interval in $\mathbb{I}$ the set of proposition letters that hold over it.
Given
an interval $I\in\mathbb{I}$, the truth of an $\HS$ formula over $I$ is inductively defined as follows (the Boolean connectives are treated as usual):
\begin{itemize}
    \item $\mathpzc{A},I\models p$ if and only if $p\in \sigma(I)$, for any $p\in\mathpzc{AP}$;
    \item $\mathpzc{A},I\models \langle X\rangle\psi$, for $X \in\{ B,E\}$, if and only if there exists $J\in\mathbb{I}$ such that $I\, X_\mathbb{I}\, J$ and $\mathpzc{A},J\models \psi$;
    \item $\mathpzc{A},I\models \langle \overline{X}\rangle\psi$, for $\overline{X} \in\{\overline{B},\overline{E}\}$, if and only if there exists $J\in\mathbb{I}$ such that $J\, X_\mathbb{I}\, I$ and $\mathpzc{A},J\models \psi$.
\end{itemize}

The next definition shows how to derive an abstract interval model from an interval structure.

\begin{definition}[Abstract interval model induced by an interval structure] 
An interval structure 
$\mathpzc{IS}=(\mathbb{S},\sigma)$, with $\mathbb{S}=(X,<)$,  \emph{induces} the abstract interval model
$ \mathpzc{A}_{\mathpzc{IS}}=(\mathpzc{AP},\mathbb{I}(\mathbb{S}), B_{\mathbb{I}(\mathbb{S})},E_{\mathbb{I}(\mathbb{S})},$ $\sigma)$,
where
 $[x,y] \, B_{\mathbb{I}(\mathbb{S})}\, [v,z]$ iff $x=v$ and $z<y$, and
 $[x,y] \, E_{\mathbb{I}(\mathbb{S})}\, [v,z]$ iff $y=z$ and $x<v$.
\end{definition}
For an interval $I$ and an $\HS$ formula $\psi$, we write $\mathpzc{IS},I\models \psi$ to mean that $\mathpzc{A}_{\mathpzc{IS}},I\models \psi$.

\subsection{Three semantic variants of $\HS$ for MC}\label{sect:3sem}

In this section we define the three variants of $\HS$ semantics $\HS_\stat$ (state-based), $\HS_\LinearPast$ (computation-tree-based), and $\HS_\LinearTime$ (trace-based) for model checking $\HS$ formulas against Kripke structures.
For each variant, the related (finite) MC problem consists of deciding whether or not a finite Kripke structure is a model of an $\HS$ formula under such a semantic variant.

%

Let us start with the \emph{state-based variant}~\cite{DBLP:conf/time/MontanariMPP14,DBLP:journals/acta/MolinariMMPP16},
where 
an abstract interval model 
is naturally associated with a given Kripke structure $\mathpzc{K}$ 
by considering the set of intervals as the set 
$\Trk_\mathpzc{K}$
of traces of $\mathpzc{K}$. 

\begin{definition}[Abstract interval model induced by a Kripke structure]\label{def:inducedmodel}
The \emph{abstract interval model induced by a  Kripke structure} $\mathpzc{K}=(\mathpzc{AP},S,\delta,\mu,s_0)$ is
$\mathpzc{A}_\mathpzc{K}=(\mathpzc{AP},\mathbb{I},B_\mathbb{I},E_\mathbb{I},\sigma)$, where
    $\mathbb{I}=\Trk_\mathpzc{K}$,
    $B_\mathbb{I}=\{(\rho,\rho')\in\mathbb{I}\times\mathbb{I}\mid \rho'\in\Pref(\rho)\}$,
    $E_\mathbb{I}=\{(\rho,\rho')\in\mathbb{I}\times\mathbb{I}\mid \rho'\in\Suff(\rho)\}$, and
    $\sigma:\mathbb{I}\to 2^\mathpzc{AP}$ is such that $\sigma(\rho)=\bigcap_{s\in\states(\rho)}\mu(s)$, for all $\rho\in\mathbb{I}$.
\end{definition}
\noindent
According to the definition of $\sigma$,
$p\in\mathpzc{AP}$ holds over $\rho=s_1\cdots s_n$
if and only if it holds over all the states $s_1, \ldots , s_n$ of $\rho$. This conforms to the \emph{homogeneity principle}, according to which a proposition letter holds over an interval
if and only if it holds over all its subintervals~\cite{roe80}.

\begin{definition}[State-based $\HS$---$\HS_\stat$]
Let $\mathpzc{K}$ be a Kripke structure and
$\psi$ be an $\HS$ formula.
A trace $\rho\in\Trk_\mathpzc{K}$ satisfies $\psi$ under the state-based semantic variant,
denoted as $\mathpzc{K},\rho\models_\stat \psi$, if it holds that $\mathpzc{A}_\mathpzc{K},\rho\models \psi$.
Moreover,
\emph{$\mathpzc{K}$ is a model of $\psi$ under the state-based semantic variant}, denoted as $\mathpzc{K}\models_\stat \psi$, if
for all \emph{initial} traces $\rho\in\Trk_\mathpzc{K}^{0}$, it holds that $\mathpzc{K},\rho\models_\stat \psi$.
\end{definition}

We now introduce the \emph{computation-tree-based semantic variant}, where we simply consider the abstract interval model \emph{induced by the computation tree} of the Kripke structure. 
Notice that since each state in a computation tree has a unique predecessor (with the exception of the initial state), this $\HS$ variant enforces a linear reference in the past.

\begin{definition}[Computation-tree-based $\HS$---$\HS_\LinearPast$]
A Kripke structure \emph{$\mathpzc{K}$ is a model of an $\HS$ formula $\psi$ under the computation-tree-based semantic variant}, written $\mathpzc{K}\models_\LinearPast \psi$, if
$\mathpzc{C}(\mathpzc{K})\models_\stat \psi$.
\end{definition}

Finally, we define the \emph{trace-based semantic variant}, which exploits the interval structures induced by the infinite paths of the Kripke structure.

\begin{definition}[Interval structure induced by an infinite path]\label{def:inducedPathIntervalStructure}
For a Kripke structure $\mathpzc{K}=(\mathpzc{AP},S,\delta,\mu,s_0)$ and an infinite path $\pi=\pi(0)\pi(1)\cdots$ of $\mathpzc{K}$,
the \emph{interval structure induced by $\pi$} is
$\mathpzc{IS}_{\mathpzc{K},\pi}=((\Nat,<),\sigma)$, where
 for each interval $[i,j]$,   $\sigma([i,j])=\bigcap_{h=i}^{j}\mu(\pi(h))$.
\end{definition}

\begin{definition}[Trace-based $\HS$---$\HS_\LinearTime$]
A Kripke structure $\mathpzc{K}$ \emph{is a model of an $\HS$ formula $\psi$ under the trace-based semantic variant}, denoted as $\mathpzc{K}\models_\LinearTime \psi$, if and only if
for each initial infinite path $\pi$  and for each initial interval $[0,i]$, it holds that  $\mathpzc{IS}_{\mathpzc{K},\pi},[0,i]\models \psi$. 
\end{definition}


In the next sections, we compare the expressiveness of the logics $\HS_{\stat}$, $\HS_{\LinearPast}$, $\HS_{\LinearTime}$, $\LTL$, $\CTL$, and $\CTLStar$ when interpreted over finite Kripke structures.
Given two logics $L_1$ and $L_2$, and two formulas $\varphi_1\in L_1$ and $\varphi_2\in L_2$, we say that   $\varphi_1$ in $L_1$ is \emph{equivalent} to  $\varphi_2$ in $L_2$ if, for every finite Kripke structure   $\mathpzc{K}$, 
$\mathpzc{K}$ is a model of $\varphi_1$ in $L_1$ if and only if $\mathpzc{K}$ is a model of $\varphi_2$ in $L_2$. 
We say that \emph{$L_2$ is subsumed by $L_1$}, denoted as $L_1\geq L_2$, if for each formula $\varphi_2\in L_2$, there exists a formula $\varphi_1\in L_1$ such that $\varphi_1$ in $L_1$ is equivalent to $\varphi_2$
in $L_2$. Moreover $L_1$ is \emph{as expressive as} $L_2$ (or $L_1$  and $L_2$ have \emph{the same expressive power}), written $L_1\equiv L_2$, if both $L_1\geq L_2$ and $L_2\geq L_1$. We say that  $L_1$ is \emph{(strictly) more expressive than} $L_2$ if $L_1\geq L_2$ and $L_2\not\geq L_1$. Finally $L_1$ and $L_2$ are \emph{expressively incomparable} if both $L_1\not\geq L_2$ and $L_2\not\geq L_1$.


\subsection{An example: a vending machine}\label{subs:vendingMach}
In this section, we give an example highlighting the differences among the $\HS$ semantic variants $\HS_\stat$, $\HS_\LinearPast$, and $\HS_\LinearTime$.

\begin{figure}[t]
\centering
\scalebox{0.8}{
\begin{tikzpicture}[->,>=stealth',shorten >=1pt,auto,node distance=2.8cm,semithick]
  \tikzstyle{every state}=[inner sep=1pt, outer sep=0pt, minimum size = 40pt]

  \node[accepting,state] (A0)                    {$\stackrel{s_0}{p_\text{\$=0}}$};
  \node[state]         (B2) [below of=A0] {$\stackrel{s_2}{p_\text{\$=2}}$};
  \node[state]         (B1) [left of=B2] {$\stackrel{s_1}{p_\text{\$=1}}$};
  \node[state]         (B05) [right of=B2] {$\stackrel{s_3}{p_\text{\$=0.50}}$};
  \node[state]         (C1) [below of=B1] {$\stackrel{s_4}{p_\text{candy}}$};
  \node[state]         (C2) [below of=B2] {$\stackrel{s_5}{p_\text{hotdog}}$};
  \node[state]         (C05) [below of=B05] {$\stackrel{s_6}{p_\text{water}}$};
  \node[state]         (D) [below of=C2] {$\stackrel{s_7}{p_\text{change}}$};
  
  \node[state]         (M1) [right of=C05] {$\stackrel{s_8}{p_\text{maint}}$};
  \node[state]         (M2) [right of=B05] {$\stackrel{s_9}{p_\text{maint\_end}}$};
  
  \path (A0) edge  [swap]  node {ins\_\$2} (B2)
            edge  [bend right,swap]  node {ins\_\$1} (B1)
            edge  [bend left,swap]  node {ins\_\$0.50} (B05)
        (B2) edge  [near start]  node {sel} (C2)
             edge  [near end]  node {sel} (C1)
             edge   node {sel} (C05)
        (B1) edge   node {sel} (C1)
             edge  [very near start]  node {sel} (C05)
        (B05) edge  [near end]  node {sel} (C05)
        (C1) edge   [bend right,very near start]  node {dispensed} (D)
        (C2) edge  [swap]  node {dispensed} (D)
        (C05) edge  [bend left,swap]  node {dispensed} (D)
        (D) edge   [bend right,swap]  node {change\_given} (M1)
        (D) edge   [out=200,in=160,looseness=1.7]  node {change\_given} (A0)
        (M1) edge  [near end]  node {maint\_ongoing} (M2)
        (M2) edge [bend left] node {maint\_failed} (M1)
        (M2) edge   [bend right,swap]  node {maint\_success} (A0);
        
\draw (-3.6,1) [dashed] rectangle (3.6,-9.5);
\draw (4.5,-1.5) [dashed] rectangle (6.8,-7);
\node (P1) at (4.3,1) {$p_\text{operative}$};
\node (P2) at (7,-1.3) {$\neg p_\text{operative}$};
\end{tikzpicture}
}
\vspace{-1cm}
\caption{Kripke structure representing a vending machine.}\label{fig:vending}
\end{figure}
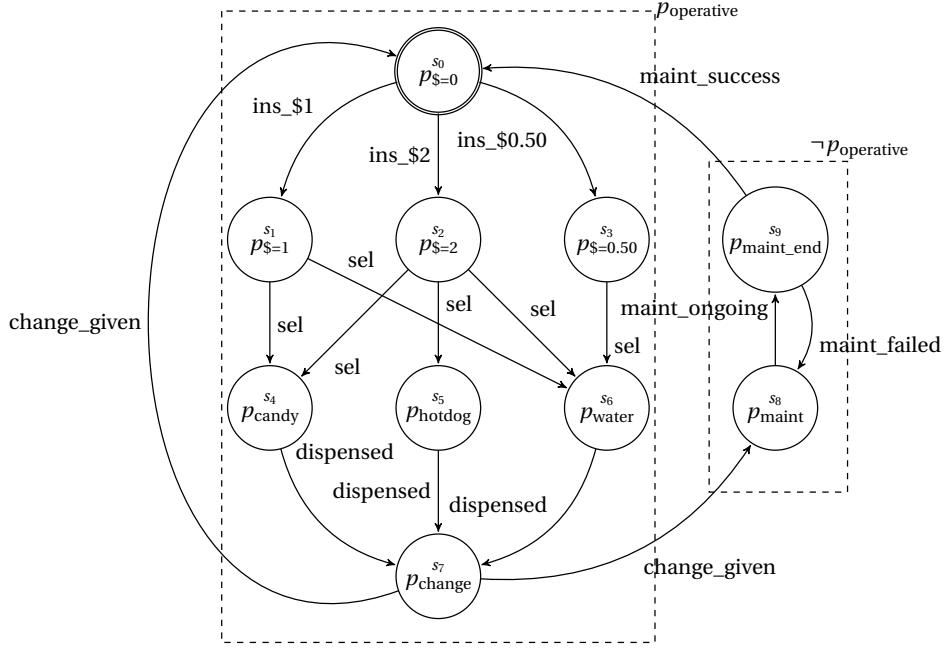

The Kripke structure of Figure~\ref{fig:vending} represents a \emph{vending machine}, which can dispense water, hot dogs, and candies.
In state $s_0$ (the initial one), no coin has been inserted into the machine (hence, the proposition letter $p_\text{\$=0}$ holds there).
Three edges, labelled by \lq\lq ins\_\$1\rq\rq, \lq\lq ins\_\$2\rq\rq, and \lq\lq ins\_\$0.50\rq\rq{}, connect $s_0$ to $s_1$, $s_2$, and $s_3$, respectively. 
Edge labels do not convey semantic value (they are neither part of the structure definition nor associated with proposition letters) and are simply used for an easy reference to edges. 
In $s_1$ (resp., $s_2$, $s_3$) the proposition letter $p_\text{\$=1}$ (resp., $p_\text{\$=2}$, $p_\text{\$=0.50}$) holds, representing the fact that 1 Dollar (resp., 2, 0.50 Dollars) has been inserted into the machine.
The cost of a bottle of water (resp., a candy, a hot dog) is \$0.50 (resp., \$1, \$2). A state $s_i$, for $i=1,2,3$, is connected to a state $s_j$, for $j=4,5,6$, only if the available credit allows one to buy the corresponding item.
Then, edges labelled by \lq\lq dispensed\rq\rq{} connect $s_4,s_5$, and $s_6$ to $s_7$. In $s_7$, the machine gives change, and  can nondeterministically move back to $s_0$ (ready for dispensing another item), or to $s_8$, where it begins an automatic maintenance activity ($p_\text{maint}$ holds there). Afterwards, state $s_9$ is reached, where maintenance ends. From there, if the maintenance activity fails (edge \lq\lq maint\_failed\rq\rq ), $s_8$ is reached again (another maintenance cycle is attempted); otherwise, maintenance concludes successfully (\lq\lq maint\_success\rq\rq ) and $s_0$ is reached. Since the machine is operating in states $s_0,\ldots,s_7$, and under maintenance in $s_8$ and $s_9$, $p_\text{operative}$ holds over the former, and it does not on the latter.

In the following, we will make use of the $\mathsf{B}$ formulas $\Length_{n}$, with $n\geq 1$: for any given $n$,
$\Length_{n}$ characterizes the intervals of length $n$, and is defined as follows:
\[\Length_n:= (\underbrace{\hsB\ldots \hsB}_{n-1} \top)\,\,\wedge\,\,
(\underbrace{[B]\ldots [B]}_{n} \bot).\]

We now give some examples of properties we can formalize under all, or some, of the $\HS$ semantic variants $\HS_\stat$, $\HS_\LinearPast$, and $\HS_\LinearTime$.

\begin{itemize}
    \item In any run of length 50, during which the machine never enters maintenance mode, it dispenses at least a hotdog, a bottle of water and a candy. 
    \[
        \Ku\not\models (p_\text{operative} \wedge \Length_{= 50})\longrightarrow \big((\hsB\hsE p_\text{hotdog}) \wedge  (\hsB\hsE p_\text{water}) \wedge (\hsB\hsE p_\text{candy}) \big)
    \]
    Clearly this property is false, as the machine can possibly dispense only one or two kinds of items.
    We start by observing that the above formula is equivalent in all of the three semantic variants of $\HS$: since modalities $\hsB$ and $\hsE$ only allow one to \lq\lq move\rq\rq{} from an interval to its subintervals, $\B\E_\LinearTime$, $\B\E_\stat$, and $\B\E_\LinearPast$ coincide (for this reason, we have omitted the subscript from the symbol $\models$). 
    Homogeneity plays a fundamental role here: asking $p_\text{operative}$ to be true implies that such a letter is true along the whole trace (thus $s_8$ and $s_9$ are always avoided).
    
    It is worth observing that the same property can be expressed in $\LTL$, for instance as follows:
    \[
    \smashoperator{\bigwedge_{i\in\{0,\ldots,49\}}}\Next^i p_\text{operative} \wedge\smashoperator[r]{\bigvee_{i,j,k\in\{1,\ldots,48\}, i\neq j\neq k\neq i}} (\Next^i p_\text{hotdog}) \wedge (\Next^j p_\text{water}) \wedge (\Next^k p_\text{candy}).
    \]
    The length of this $\LTL$ formula is \emph{exponential} in the number of items (in this case, 3), whereas the length of the above $\HS$ one is only linear. As a matter of fact, we will prove (Theorem~\ref{theo:succinctnessHSlin}) that $\B\E$ is at least exponentially more succinct than $\LTL$.
    
    \item If the credit is \$0.50, then no hot dog or candy may be provided.
    \[
        \Ku\models (\hsE p_\text{\$=0.50})\longrightarrow \neg \hsA (\Length_{=2} \wedge \hsE (p_\text{hotdog} \vee p_\text{candy})) 
    \]
    We observe that a trace satisfies $\hsE p_\text{\$=0.50}$ if and only if it ends in $s_3$.
    This property is satisfied under all of the three semantic variants, even though the nature of future differs among them (recall Figure~\ref{fig:ST}, \ref{fig:CT}, and \ref{fig:LN}). As we have already mentioned, a linear setting (rather than branching) is suitable for the specification  of dynamic behaviors, because it considers states \emph{of a computation}; conversely, a branching approach focuses on machine states (and thus on the structure of a system).
    
    In this case, only the state $s_6$ can be reached from $s_3$, regardless of the nature of future. For this reason, $\HS_\stat$, $\HS_\LinearPast$, and $\HS_\LinearTime$ behave in the same way. 
    
    \item Let us exemplify now a difference between $\HS_\stat$ (and $\HS_\LinearPast$) and $\HS_\LinearTime$.
    \[
        \begin{array}{l}
            \Ku\models_\stat \\
            \Ku\models_\LinearPast  \\
            \Ku\not\models_\LinearTime
        \end{array}
        (\hsE p_\text{maint\_end})\longrightarrow \hsA\hsE p_\text{operative}
    \]
    This is a structural property, requiring that when the machine enters state $s_9$ (where maintenance ends), it can become again operative reaching state $s_0$ ($s_9$ is not a lock state for the system). This is clearly true when future is branching and it is not when future is linear: $\HS_\LinearTime$ refers to system computations, and some of these may ultimately loop between $s_8$ and $s_9$.
    
    \item 
    Conversely, some properties make sense only if they are predicated over computations. This is the case, for instance, of fairness.
    \[
        \begin{array}{l}
            \Ku\models_\stat \\
            \Ku\models_\LinearPast  \\
            \Ku\not\models_\LinearTime
        \end{array}
        ([A]\hsA\hsE p_\text{maint})\longrightarrow [A]\hsA\hsE p_\text{operative}
    \]
    Assuming the trace-based semantics, the property requires that if a system computation enters infinitely often into maintenance mode, it will infinitely often enter operation mode.
    Again, this is not true, as some system computations may ultimately loop between $s_8$ and $s_9$ (hence, they are not fair). On the contrary, such a property is trivially true under $\HS_\stat$ or $\HS_\LinearPast$, as, for any initial trace $\rho$, it holds that $\Ku,\rho\models \hsA\hsE p_\text{operative}$.
    
    \item We conclude with a property showing the difference between linear and branching \emph{past}, that is, between $\HS_\stat$ and $\HS_\LinearTime$ (and $\HS_\LinearPast$).
    The requirement is the following: the machine may dispense water with any amount of (positive) credit.
    \[
        \begin{array}{l}
            \Ku\models_\stat \\
            \Ku\not\models_\LinearPast  \\
            \Ku\not\models_\LinearTime
        \end{array}
         (\hsE p_\text{water})\longrightarrow \hsE\big(p_\text{water}\wedge\smashoperator[r]{\bigwedge_{p\in\{p_\text{\$=2},p_\text{\$=1},p_\text{\$=0.50}\}}}\hsAt(\Length_{=2} \wedge\hsB p)\big)
    \]
    Again, this one is a structural property, that cannot be expressed in $\HS_\LinearTime$ or $\HS_\LinearPast$, as these refer to a specific computation in the past. Conversely, it is true under $\HS_\stat$, since $s_6$ is backward reachable in one step by $s_1$, $s_2$, and $s_3$.
\end{itemize}

\section{Equivalence between $\LTL$ and $\HS_\LinearTime$}\label{sec:CharacterizezionOfLeniarTimeHS}

In this section, we show that $\HS_\LinearTime$ is as expressive as $\LTL$ even for small syntactical fragments of $\HS_\LinearTime$. To this end, we exploit the well-known equivalence between $\LTL$ and the first-order fragment of monadic second-order logic over infinite words ($\FO$ for short). Recall that, given a countable set $\{x,y,z,\ldots\}$ of (position) variables, the $\FO$ formulas $\varphi$ over a set of proposition letters $\mathpzc{AP}=\{p,\ldots \}$
are defined as:
\[
\varphi ::= \top \ |\ p(x)    \ |\ x\leq y \ |\ x <y  \ |\ \neg\,\varphi \ |\ \varphi\, \wedge\, \varphi
 \ |\ \exists x. \varphi\; .
 \]

We interpret $\FO$ formulas $\varphi$ over infinite paths $\pi$ of Kripke structures $\mathpzc{K}=(\mathpzc{AP},S,\delta,\mu,s_0)$. 
Given a variable \emph{valuation} $g$, assigning to each variable  a position $i\geq 0$,  the satisfaction relation
$(\pi,g)\models \varphi$ corresponds to the standard satisfaction relation $(\mu(\pi),g)\models \varphi$, where $\mu(\pi)$ is the infinite word over $2^{\mathpzc{AP}}$ given by $\mu(\pi(0))\mu(\pi(1))\cdots$. 
More precisely, 
$(\pi,g)\models \varphi$ is inductively defined as follows
 (we omit the standard rules for the Boolean connectives):
\[
\begin{array}{ll}
 (\pi,g)\models p(x) & \Leftrightarrow  p \in \mu(\pi(g(x))),\\
 (\pi,g)\models x\ op\ y & \Leftrightarrow  g(x)\ op\ g(y), \mbox{ for } op \in \{<, \leq\},\\
   (\pi,g)\models \exists x. \varphi & \Leftrightarrow  (\pi,g[x\leftarrow i])\models \varphi \text{ for some } i\geq 0 ,
\end{array}
\]
where $g[x\leftarrow i](x)=i$ and $g[x\leftarrow i](y)=g(y)$ for $y\neq x$. Note that the satisfaction relation  depends only on the values assigned to the variables occurring free in the given formula $\varphi$.
We write $\pi\models \varphi$ to mean that $(\pi,g_0)\models \varphi$, where $g_0(x)=0$ for each variable $x$. An $\FO$ sentence is a formula with no free variables.  The following is a well-known result (Kamp's theorem~\cite{Kamp}).

\begin{proposition}\label{theo:FromFOtoLTL} Given an $\FO$ sentence $\varphi$ over $\mathpzc{AP}$, one can construct an $\LTL$ formula $\psi$ such that, for all Kripke structures $\mathpzc{K}$ over $\mathpzc{AP}$ and infinite paths $\pi$, it holds that
$\pi \models \varphi$ if and only if $\pi,0\models \psi$.
\end{proposition}

Given a $\HS_\LinearTime$ formula $\psi$, we now construct  an 
$\FO$ sentence $\psi_\FO$ such that, for all Kripke structures $\mathpzc{K}$, 
$\mathpzc{K}\models_\LinearTime \psi$ if and only if for each initial infinite path $\pi$ of $\mathpzc{K}$, $\pi\models \psi_\FO$.

We start by defining a mapping $h$ assigning to each triple $(\varphi,x,y)$, consisting of a $\HS$ formula $\varphi$ and two distinct position variables $x,y$, an $\FO$ formula having as free variables $x$ and $y$.
The mapping $h$ returns the $\FO$ formula defining the semantics of the $\HS$ formula $\varphi$ interpreted over an interval bounded by the positions $x$ and $y$. 
 
 The function $h$ is homomorphic with respect to the Boolean connectives, and is defined for proposition letters and modal operators as follows (here $z$ is a fresh position variable):
\[ \begin{array}{ll}
h(p,x,y) & = \forall z.((z\geq x \wedge z\leq y) \rightarrow p(z)),\\
h(\hsE\psi,x,y) & = \exists z.(z> x \wedge z\leq y \wedge h(\psi,z,y)),\\
h(\hsB\psi,x,y) & = \exists z.(z\geq x  \wedge z< y \wedge h(\psi,x,z)),\\
h(\hsEt\psi,x,y) & = \exists z.(z< x \wedge h(\psi,z,y)),\\
h(\hsBt\psi,x,y) & = \exists z.(z> y  \wedge  h(\psi,x,z)).
\end{array} \]
It is worth noting that homogeneity plays a crucial role in the definition of $h(p,x,y)$ (without it, a binary predicate would be necessary to encode the truth of $p$ over $[x,y]$).

Given a Kripke structure $\mathpzc{K}$, an infinite path $\pi$, an interval of positions $[i,j]$, and an $\HS_\LinearTime$ formula $\psi$, by a straightforward induction on the structure of
$\psi$, we can show that 
$\mathpzc{IS}_{\mathpzc{K},\pi},[i,j]\models \psi$
if and only if $(\pi,g)\models h(\psi,x,y)$ for any valuation such that $g(x)= i$ and $g(y)=j$. 

Now, let us consider the
$\FO$ sentence $h(\psi)$ given by $\exists x ((\forall z.  z\geq x)\wedge \forall y.  h(\psi,x,y))$. Clearly $\mathpzc{K}\models_\LinearTime \psi$ if and only if for each initial infinite path $\pi$ of $\mathpzc{K}$, $\pi\models h(\psi)$.  By Proposition~\ref{theo:FromFOtoLTL}, it follows that one can construct an $\LTL$ formula $h'(\psi)$ such that $h'(\psi)$ in $\LTL$ is equivalent to $\psi$ in $\HS_\LinearTime$. 
Thus, we obtain the following expressiveness containment.
 
\begin{theorem}\label{theo:HSTracedBasedIsINLTL} $\LTL\geq \HS_\LinearTime$.
\end{theorem}

Now we show that also the converse containment holds, that is, $\LTL$ can be translated into $\HS_\LinearTime$.
Actually, it is worth noting that for such a purpose the fragment $\A\B$ of $\HS_\LinearTime$, featuring only modalities for $A$ and $B$, is expressive enough.

\begin{theorem}\label{theo:FromLTLTOSmallFragments} Given an $\LTL$ formula $\varphi$, one can construct in linear-time an $\mathsf{AB}$  formula $\psi$ such that $\varphi$ in $\LTL$ is equivalent to $\psi$ in $\mathsf{AB}_\LinearTime$.
\end{theorem}
\begin{proof}
Let $f: \LTL \to \mathsf{AB}$ be the mapping, homomorphic with respect to the Boolean connectives, defined as follows:  
\begin{gather*}
f(p)=p, \text{ for each proposition letter } p,\\
f(\Next\psi) = \hsA (\Length_{2}\wedge \hsA(\Length_1\wedge f(\psi))),\\
f(\psi_1\until\psi_2) = \hsA \bigl(\hsA(\Length_{1}\wedge f(\psi_2))\wedge[B](\hsA(\Length_{1}\wedge f(\psi_1))\bigr).
\end{gather*}

Given a Kripke structure $\mathpzc{K}$, an infinite path $\pi$, a position $i\geq 0$, and an $\LTL$ formula $\psi$, by a straightforward induction on the structure of
$\psi$ we can show that $\pi,i\models \psi$ if and only if
$\mathpzc{IS}_{\mathpzc{K},\pi},[i,i]\models f(\psi)$.
Hence $\mathpzc{K}\models \psi$ if and only if $\mathpzc{K}\models_\LinearTime \Length_1\rightarrow f(\psi)$.
\end{proof}

The next corollary follows immediately from Theorem~\ref{theo:HSTracedBasedIsINLTL} and Theorem~\ref{theo:FromLTLTOSmallFragments}.
\begin{corollary}\label{cor:HSLinearCharacterization} $\HS_\LinearTime$ and $\LTL$ have the same expressive power.
\end{corollary}



While there is no difference in the expressive power between $\LTL$ and $\HS_\LinearTime$, things change if we consider succinctness. Whereas Theorem~\ref{theo:FromLTLTOSmallFragments} shows that it is possible to convert any $\LTL$ formula into an equivalent $\HS_\LinearTime$ one in linear time, the following theorem holds. 

\begin{theorem}\label{theo:succinctnessHSlin} $\HS_\LinearTime$ is at least exponentially more succinct than $\LTL$.
\end{theorem}
\begin{proof}
To prove the statement, it suffices to provide an $\HS_\LinearTime$ formula $\psi$ for which there exists no $\LTL$ equivalent formula whose size is polynomial in $|\psi|$.

To this end, we restrict our attention to the fragment $\B\E_\LinearTime$. Since modalities $\hsB$ and $\hsE$ only allow one to \lq move\rq{} from an interval to its subintervals, $\B\E_\LinearTime$ actually coincides with $\B\E_\stat$, whose MC is known to be hard for $\EXPSPACE$~\cite{BMMPS16}. Thus, in particular, it is possible to encode by means of a $\B\E_\LinearTime$ formula $\psi_{\text{cpt}}$ the (unique) computation of a deterministic Turing machine using $b(n)\in O(2^n)$ bits that, when executed on input $0^n$, for some natural number $n\geq 1$, counts in binary from $0$ to $2^{2^n}-1$, by repeatedly summing 1, and finally accepts. The length of $\psi_{\text{cpt}}$ is \emph{polynomial} in $n$, and the unique trace which satisfies it (that is, that encodes such a computation) has length $\ell(n)\geq b(n)\cdot 2^{2^n}$. 


Conversely, it is known that $\LTL$ features a \emph{single-exponential small-model property}~\cite{DBLP:books/cu/Demri2016}, 
stating that, for every \emph{satisfiable} $\LTL$ formula $\varphi$, there are $u,v\in S^*$ with $|u|\leq 2^{|\varphi|}$ and $|v|\leq |\varphi|\cdot 2^{|\varphi|}$, such that $u\cdot v^\omega ,0 \models \varphi$. 
This allows us to conclude (by an easy contradiction argument) that there is no polynomial-length (w.r.t.\  $|\psi_{\text{cpt}}|$, and thus to $n$) $\LTL$ formula that can encode the aforementioned computation. An \emph{exponential-length} $\LTL$ formula would be needed for such an encoding.
\end{proof}

Exactly the same argument can be used to show that $\HS_\LinearTime$ is at least exponentially more succinct than the extension of $\LTL$ with past modalities (denoted in the following as $\LTLP$)~\cite{DBLP:journals/igpl/LichtensteinP00}.

\section{A characterization of $\HS_\LinearPast$}\label{sec:characterizationHSLinearPast}

In this section, we will focus our attention on the computation-tree-based semantic variant $\HS_\LinearPast$, showing that  it is as expressive as \emph{finitary} $\CTLStar$. As a matter of fact, the result can be proved to hold already for the syntactical fragment $\mathsf{ABE}$ which does not feature transposed modalities. In addition, we show that $\HS_\LinearPast$ is subsumed by $\CTLStar$.

\subsection{From finitary $\CTLStar$ to $\HS_\LinearPast$}
We first show that finitary $\CTLStar$ is subsumed by $\HS_\LinearPast$. 
As a preliminary fundamental step, we prove that  
when interpreted over finite words, the $\mathsf{BE}$ fragment of $\HS$ and $\LTL$ 
define the same class of finitary languages (Theorem \ref{theoremFromLTLtoBEoverFiniteWords}).

For an $\LTL$ formula $\varphi$ with proposition letters
over  an alphabet $\Sigma$ (in our case $\Sigma$ is  
 $2^{\mathpzc{AP}}$), let us denote by $\Lang_\act(\varphi)$ 
the set of non-empty finite words over $\Sigma$ satisfying $\varphi$ under the standard action-based semantics of $\LTL$, interpreted over finite words (see \cite{vardi1996automata}).
A similar notion can be given for $\mathsf{BE}$ formulas $\varphi$ with proposition letters in $\Sigma$ (under the homogeneity assumption).
Then, $\varphi$ denotes a language, written $\Lang_\act(\varphi)$, of non-empty finite words over $\Sigma$ inductively defined as:
\begin{itemize}
  \item $\Lang_\act(a)= a^{+} $, for $a\in\Sigma$ (we observe that this definition reflects the homogeneity assumption);
  \item $\Lang_\act(\neg \varphi)= \Sigma^{+}\setminus \Lang_\act(\varphi)$;
  \item $\Lang_\act(\varphi_1\wedge \varphi_2)=\Lang_\act(\varphi_1)\cap  \Lang_\act(\varphi_2)$;
  \item $\Lang_\act(\hsB \varphi)= \{w\in \Sigma^{+}\mid \Pref(w)\cap \Lang_\act(\varphi)\neq \emptyset\}$;
    \item $\Lang_\act(\hsE \varphi)= \{w\in \Sigma^{+}\mid \Suff(w)\cap \Lang_\act(\varphi)\neq \emptyset\}$.
\end{itemize}

We prove that, under the action-based semantics, $\mathsf{BE}$ formulas and $\LTL$ formulas define the same class of finitary languages.

To prove that the finitary languages defined by $\LTL$ formulas are subsumed by 
those defined by $\mathsf{BE}$ formulas we exploit an algebraic condition introduced by Wilke in~\cite{stacs/Wilke99}, called \emph{$\LTL$-closure}, which gives, for a class of finitary languages, a sufficient condition to guarantee the inclusion of the class of $\LTL$-definable languages.
The converse inclusion, that is, the class of finitary languages defined by the fragment $\mathsf{BE}$ is subsumed by that defined by $\LTL$, can be proved by a technique similar to that used in Section~\ref{sec:CharacterizezionOfLeniarTimeHS}, and thus omitted.

We start by considering the former inclusion recalling from~\cite{stacs/Wilke99} a sufficient condition for a class of finitary languages to include the class of finitary languages which are $\LTL$-definable.

  \begin{definition}[\LTL-closure]\label{def:LTLClosure} A class $\mathcal{C}$  of languages of finite words over finite alphabets is \emph{$\LTL$-closed} if and only if the following conditions are satisfied, where $\Sigma$ and $\Delta$ are finite alphabets, $b\in \Sigma$ and $\Gamma=\Sigma\setminus\{b\}$:
  \begin{enumerate}
    \item $\mathcal{C}$ is closed under language complementation and language intersection;
    \item if $\Lang \in \mathcal{C}$ with $\Lang\subseteq \Gamma^{+}$, then  $\Sigma^* b\Lang$, $\Sigma^* b(\Lang+\varepsilon)$,
    $ \Lang b \Sigma^* $, $ (\Lang+\varepsilon) b \Sigma^* $ are in  $\mathcal{C}$;
    \item Let $U_0= \Gamma^{*}b$, $h_0:U_0 \rightarrow \Delta$, and $h:U_0^{+} \rightarrow \Delta^{+}$ be defined by
    $h(u_0u_1\cdots u_n)\!=\!h_0(u_0)\!\cdots\! h_0(u_n)$. Assume that for each $d\in \Delta$, the language $\Lang_d =\{u\in \Gamma^{+}\mid h_0(ub)= d\}$
    is in $\mathcal{C}$. Then, for each language $\Lang\in \mathcal{C}$ such that $\Lang\subseteq \Delta^{+}$,
    the language $\Gamma^{*}b h^{-1}(\Lang)\Gamma^{*}$ is in $\mathcal{C}$.
  \end{enumerate}
  \end{definition}
  
\begin{figure}[b]
    \centering
    \scalebox{0.6}{\includegraphics{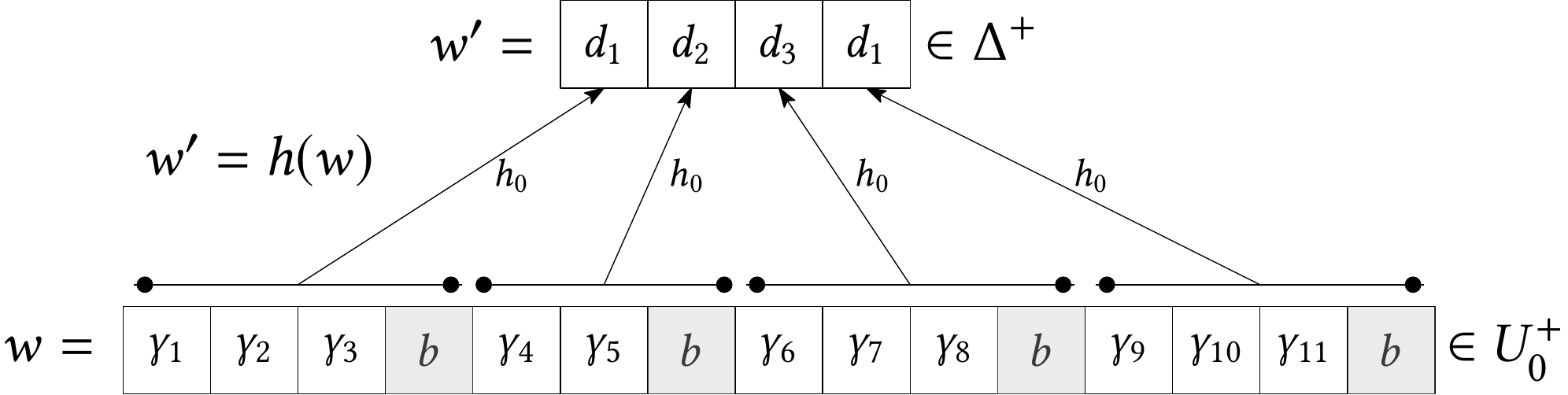}}
    \caption{ 
            Visual description of condition 3 of Definition~\ref{def:LTLClosure} (\LTL-closure). 
    }
    \label{fig:LTLclos3}
\end{figure}
  
In Figure~\ref{fig:LTLclos3}, we graphically  depict condition 3 of the definition of \LTL-closure. In the proposed example, we have:
              $(i)$~for all $i$, $d_i\in \Delta$ and $\gamma_i\in\Gamma$,
              $(ii)$~$w= (\gamma_1\gamma_2 \gamma_3 b) ( \gamma_4\gamma_5 b) ( \gamma_6\gamma_7 \gamma_8 b)( \gamma_9\gamma_{10} \gamma_{11} b)\in U_0^4$,
              $(iii)$~$w'=h(w)=h_0(\gamma_1\gamma_2 \gamma_3 b)h_0(\gamma_4\gamma_5 b)h_0(\gamma_6\gamma_7 \gamma_8 b)h_0(\gamma_9\gamma_{10} \gamma_{11} b)=d_1 d_2 d_3 d_1 \in\Delta^4$.
            For instance, $\gamma_1\gamma_2\gamma_3,\gamma_9\gamma_{10} \gamma_{11}\in L_{d_1}$ and $\gamma_4\gamma_5\in L_{d_2}$.

The following result holds~\cite{stacs/Wilke99}.
  \begin{theorem}\label{theo:CharacterizationFinitaryLTL} Any $\LTL$-closed class  $\mathcal{C}$ of finitary languages includes the class of  $\LTL$-definable finitary languages.
  \end{theorem}

Therefore, to prove that the finitary languages defined by  $\mathsf{BE}$ formulas subsume those defined by $\LTL$, as stated by Theorem~\ref{theoremFromLTLtoBEoverFiniteWords} below, it suffices to prove that \emph{the class of finitary languages definable by $\mathsf{BE}$ formulas is $\LTL$-closed}, and to apply Theorem~\ref{theo:CharacterizationFinitaryLTL}.
We observe that, by definition, the class of $\mathsf{BE}$-definable languages is obviously closed under language complementation and intersection (condition 1 of Definition~\ref{def:LTLClosure}). The fulfillment of conditions 2 and 3 of Definition~\ref{def:LTLClosure} is then proved by the two following Lemmata~\ref{lemma1:CharacterizationFinitaryLTL} and~\ref{lemma2:CharacterizationFinitaryLTL}, respectively.


 \begin{lemma}\label{lemma1:CharacterizationFinitaryLTL} Let $\Sigma$ be a finite alphabet,
  $b\in \Sigma$, $\Gamma=\Sigma\setminus\{b\}$, $\Lang\subseteq \Gamma^{+}$, and $\psi$ be a $\mathsf{BE}$ formula over $\Gamma$ such that $\Lang_\act(\psi)=\Lang$. Then,
  there are $\mathsf{BE}$ formulas defining (under the action-based semantics) the languages
   $b\Lang$, $\Sigma^* b\Lang$, $\Sigma^* b(\Lang+\varepsilon)$,
   $\Lang b$, $ \Lang b \Sigma^* $, $ (\Lang+\varepsilon) b \Sigma^* $, and $b \Lang b$.
  \end{lemma}
  \begin{proof}
We focus on the cases for the languages $b\Lang$, $\Sigma^* b\Lang$, $\Sigma^* b$, and  $b \Lang b$ (for the other languages, the proof is similar: $\Sigma^*b(\Lang + \varepsilon)=\Sigma^* b\Lang+ \Sigma^* b$, $\Lang b$ is symmetric to $b\Lang$, $\Lang b\Sigma^*$ to $\Sigma^* b\Lang$, and $ (\Lang+\varepsilon) b \Sigma^* $ to $\Sigma^*b(\Lang + \varepsilon)$). Let $\psi$ be a $\mathsf{BE}$ formula over $\Gamma$ such that $\Lang_\act(\psi)=\Lang$. 

\paragraph{Language $b\Lang$.} The $\mathsf{BE}$ formula defining the language $b\Lang$ is the formula:
\begin{equation}
\label{eq-bL}
(\neg \Length_1 \wedge \hsB b \wedge [E](\neg b \wedge [B]\neg b)) \wedge h_b(\psi),
\end{equation}
where the  formula  $h_b(\psi)$ is inductively defined on the structure of $\psi$ in the following way. The mapping $h_b$ is homomorphic with respect to the Boolean connectives, while for the atomic actions in $\Gamma$ and the modalities $\hsE$ and $\hsB$, it is defined  as follows:
 \begin{itemize}
   \item for all $a\in \Gamma$, $h_b(a)= a\vee (\hsB b \wedge \hsE a \wedge [E]a)$;
   \item $h_b(\hsB\theta)= (\hsB h_b(\theta) \wedge \neg\hsB b) \vee \hsB(h_b(\theta)\wedge \hsB b)$;
   \item $h_b(\hsE\theta)= (\hsE h_b(\theta) \wedge \neg\hsB b) \vee (\hsB b \wedge \hsE \hsE h_b(\theta))$.
 \end{itemize}
The first conjunct of the formula of (\ref{eq-bL}) ensures that a word $u'$ in the defined language has length at least $2$ and it has the form $bu$ without any occurrence of $b$ in $u$. The second conjunct $h_b(\psi)$ ensures that $u$ belongs to the language defined by $\psi$. For atomic actions and temporal modalities, $h_b(\psi)$ is a disjunction of two possible choices; the appropriate one is forced at top level by the first conjunct of the formula of (\ref{eq-bL}), that constrains one and only one $b$ to occur in the word in the first position.

By a straightforward structural induction on $\psi$, it can be shown that the following fact holds.

\emph{Claim 1.} Let $u\in \Gamma^{+}$, $u'= bu$, and $|u|= n+1$. Then, for all $i,j \in [0,n]$ with $i\leq j$, $u[i,j] \in \Lang_\act(\psi)$
if and only if $u'[\,\hat{i},j+1]\in \Lang_\act(h_b(\psi))$, where  $\hat{i}= i$ if $i=0$, and $\hat{i}=i+1$ otherwise.

 By Claim 1, for each $u\in \Gamma^{+}$, $u\in \Lang_\act(\psi)$ if and only if $bu\in \Lang_\act(h_b(\psi))$. Therefore, the formula of (\ref{eq-bL}) captures the language $b\Lang_\act(\psi)$. 

\paragraph{Languages $\Sigma^* b\Lang$ and $\Sigma^* b$.} Following the proof given for the case of the language $b\Lang$, with $\Lang\subseteq \Gamma^{+}$, one can construct a
 $\mathsf{BE}$ formula  $\varphi$ defining the language $b\Lang$. Hence, the $\mathsf{BE}$ formula $\varphi\vee \hsE\varphi$ defines $\Sigma^* b\Lang$. The $\mathsf{BE}$ formula defining $\Sigma^* b$ is $b \vee \hsE b$. 

\paragraph{Language $b\Lang b$.} By the proof given for the language $b\Lang$, with $\Lang\subseteq \Gamma^{+}$, one can build a $\mathsf{BE}$ formula  $\varphi$ defining the language $b\Lang$. The $\mathsf{BE}$ formula defining the language $b\Lang b$ is the formula:
\begin{equation}
\label{eq-bLb}
(\neg \Length_1 \wedge \neg \Length_2 \wedge \hsB b \wedge \hsE b \wedge [E][B]\neg b) \wedge k_b(\varphi)
\end{equation}
where the formula  $k_b(\varphi)$ is inductively defined on the structure of $\varphi$  in the following way. 
The mapping $k_b$ is homomorphic with respect to the Boolean connectives, while for the atomic actions in $\Sigma$ and the modalities $\hsE$ and $\hsB$, it is defined  as follows:
 \begin{itemize}
   \item for all $a\in \Gamma$, $k_b(a)= a\vee (\hsE b \wedge \hsB a \wedge [B]a)$;
   \item $k_b(b)=b$;
   \item $k_b(\hsB\theta)= (\hsB k_b(\theta) \wedge \neg\hsE b) \vee (\hsE b \wedge \hsB \hsB k_b(\theta))$.
      \item $k_b(\hsE\theta)= (\hsE k_b(\theta)\wedge \neg\hsE b) \vee \hsE(k_b(\theta)\wedge \hsE b)$.
 \end{itemize}

The first conjunct of the formula of (\ref{eq-bLb}) ensures that a word $u'$ in the defined language has length at least $3$ and it has the form $bub$ without any occurrence of $b$ in $u$. The second conjunct $k_b(\varphi)$ ensures that $bu$ belongs to the language defined by $\varphi$. Similarly to the case of the language $b\Lang$, for atomic actions (different from $b$) and temporal modalities, $k_b(\psi)$ is a disjunction of two possible choices; the appropriate one is forced at top level by the first conjunct of the formula of (\ref{eq-bLb}), that constrains one and only one $b$ to occur in the word in the last position.

By a straightforward structural induction on $\varphi$, it can be shown that the following fact holds.

\emph{Claim 2.} Let $u\in \Gamma^{+}$ and $|bu|= n+1$. Then, for all $i,j \in [0,n]$ with $i\leq j$, $bu[i,j] \in \Lang_\act(\varphi)$
if and only if $bub[i,\hat{j}]\in \Lang_\act(k_b(\varphi))$ where  $\hat{j}= j$ if $j<n$, and $\hat{j}=n+1$ otherwise. 

 By Claim 2, for each $u\in \Gamma^{+}$, $bu\in \Lang_\act(\varphi)$ if and only if $bub\in \Lang_\act(k_b(\varphi))$ implying that the formula of (\ref{eq-bLb})
 defines the language $\Lang_\act(\varphi)b$. 
 This concludes the proof of the lemma.
  \end{proof}

  \begin{lemma}\label{lemma2:CharacterizationFinitaryLTL}Let $\Sigma$ and $\Delta$ be finite alphabets,
  $b\in \Sigma$, $\Gamma=\Sigma\setminus\{b\}$, $U_0= \Gamma^{*}b$, $h_0:U_0 \rightarrow \Delta$ and $h:U_0^{+} \rightarrow \Delta^{+}$ be defined by
    $h(u_0u_1\cdots u_n)=h_0(u_0)\cdots h_0(u_n)$. Assume that, for each $d\in \Delta$, there is a $\mathsf{BE}$ formula capturing the language $\Lang_d =\{u\in \Gamma^{+}\mid h_0(ub)= d\}$.  Then, for each $\mathsf{BE}$ formula $\varphi$ over $\Delta$, one can construct a $\mathsf{BE}$ formula over $\Sigma$
    capturing
    the language $\Gamma^{*}b h^{-1}(\Lang_\act(\varphi))\Gamma^{*}$.
  \end{lemma}
  \begin{proof} By hypothesis and Lemma~\ref{lemma1:CharacterizationFinitaryLTL}, for each $d\in \Delta$ there exists a
  $\mathsf{BE}$ formula $\theta_d$ over $\Sigma$ defining the language $b \Lang_d b$, where $\Lang_d =\{u\in \Gamma^{+}\mid h_0(ub)= d\}$.
  Hence, there is a
  $\mathsf{BE}$ formula $\hat{\theta}_d$ over $\Sigma$ capturing the language $b \hat{\Lang}_d b$, where $\hat{\Lang}_d =\{u\in \Gamma^{*}\mid h_0(ub)= d\}$ (note that
 $\Lang_d= \hat{\Lang}_d\setminus\{\varepsilon\}$).

 Let $\varphi$ be a  $\mathsf{BE}$ formula over $\Delta$. By structural induction over $\varphi$, we construct a $\mathsf{BE}$ formula $\varphi^{+}$ over $\Sigma$
 such that $\Lang_\act(\varphi^{+})= \Gamma^{*}b h^{-1}(\Lang_\act(\varphi))\Gamma^{*}$.
The formula $\varphi^{+}$ is defined as follows:
\begin{itemize}
  \item $\varphi= d$ with $d\in \Delta$. We have that $\Lang_\act(d)=d^{+}$ and $\Gamma^{*}b h^{-1}(\Lang_\act(d))\Gamma^{*}$ is the set of finite words in
  $\Gamma^{*}b \Sigma^{*} b\Gamma^{*}$ such that each subword $u[i,j]$ of $u$ which is in $b\Gamma^{*}b$ is in $b\hat{\Lang}_d b$ as well. Using the formula  $\psi_b:= \neg \Length_1 \wedge \hsB b \wedge \hsE b \wedge [E][B]\neg b$ to define the language $b\Gamma^{*}b$, $\varphi^{+}$ is defined as follows:
  \[
  \varphi^{+} =(\hsG\psi_b) \wedge [G](\psi_b \rightarrow \hat{\theta}_d).
  \]
  \item $\varphi = \neg \theta$. We have that
  \[
 \Gamma^{*}b h^{-1}(\Lang_\act(\varphi))\Gamma^{*} = 
 \Gamma^{*}b h^{-1}(\Delta^+ \setminus \Lang_\act(\theta))\Gamma^{*} = 
 \Gamma^{*}b h^{-1}(\Delta^+)\Gamma^{*} \cap \overline{\Gamma^{*}b  h^{-1}(\Lang_\act(\theta))\Gamma^{*}}, 
  \]
  where $\Gamma^{*}b h^{-1}(\Delta^+)\Gamma^{*}$ restricts the set of \lq candidate\rq{} models to the well-formed ones.
  
  Thus, taking $\psi_b$ as defined in the previous case, $\varphi^{+}$ is given by: 
  \[\varphi^{+}=(\hsG\psi_b) \wedge [G](\psi_b \rightarrow \bigvee_{d\in \Delta}\hat{\theta}_d) \wedge \neg \theta^{+},\]
  where, by the inductive hypothesis, $\Lang_\act(\theta^{+})=\Gamma^{*}b  h^{-1}(\Lang_\act(\theta))\Gamma^{*}$. 
  \item $\varphi = \theta\wedge \psi$. We simply have $\varphi^{+}= \theta^{+}\wedge \psi^{+}$.
  \item $\varphi =\hsB \theta$. First, we note that $\Gamma^{*}b h^{-1}(\Lang_\act(\hsB \theta))\Gamma^{*}$ is the set of finite words in the language  $\Gamma^{*}b h^{-1}(\Lang_\act(\theta))h^{-1}(\Delta^+)\Gamma^{*}$, which is included in the
  language  $\Gamma^{*}bh^{-1}(\Delta^+)\Gamma^{*}$ defined by the formula
  $ [G](\psi_b \rightarrow \bigvee_{d\in \Delta}\hat{\theta}_d).$ Note also that,
  by the inductive hypothesis, $\Gamma^{*}b h^{-1}(\Lang_\act(\theta))$ is included in the language of $\theta^{+}$. 
  Thus, $\varphi^{+}$ is given by:
 \[
 \varphi^{+} =  [G](\psi_b \rightarrow \bigvee_{d\in \Delta}\hat{\theta}_d) \wedge (\xi\vee \hsB\xi),
 \]
 where $\xi=(\hsE b) \wedge  \hsB(\theta^{+}\wedge \hsE b)$.
 %
%
  \item $\varphi =\hsE \theta$. $\Gamma^{*}b h^{-1}(\Lang_\act(\hsE \theta))\Gamma^{*}$ is the set  $\Gamma^{*}b h^{-1}(\Delta^+) h^{-1}(\Lang_\act(\theta))\Gamma^{*}$ included in the
  language  $\Gamma^{*}bh^{-1}(\Delta^+)\Gamma^{*}$, symmetrically to the previous case.
  %
  %
  Thus, $\varphi^{+}$ is given by:
 \[
 \varphi^{+} =  [G](\psi_b \rightarrow \bigvee_{d\in \Delta}\hat{\theta}_d) \wedge (\xi'\vee \hsE\xi'),
\]
where $\xi'=(\hsB b) \wedge  \hsE(\theta^{+}\wedge \hsB b)$.\qedhere
 \end{itemize}

  \end{proof}

Since, by Lemmata~\ref{lemma1:CharacterizationFinitaryLTL} and~\ref{lemma2:CharacterizationFinitaryLTL}, the class of finitary languages definable by $\mathsf{BE}$ formulas is $\LTL$-closed, by Theorem~\ref{theo:CharacterizationFinitaryLTL} we get the following result.

\begin{theorem}\label{theoremFromLTLtoBEoverFiniteWords} Let $\varphi$ be an $\LTL$ formula over a finite alphabet $\Sigma$. Then, there exists a $\mathsf{BE}$ formula
 $\varphi_\HS$ over $\Sigma$ such that $\Lang_\act(\varphi_\HS) =\Lang_\act(\varphi)$.
\end{theorem}

 The result expressed in Theorem~\ref{theoremFromLTLtoBEoverFiniteWords} above is used to prove that finitary $\CTLStar$ is subsumed by the fragment $\mathsf{ABE}$ under the state-based semantics.

  \begin{theorem}\label{theo:FromFInitaryCTLStarToFutureHSBranching} Let $\varphi$ be a finitary $\CTLStar$ 
  formula over $\mathpzc{AP}$. Then, there is an $\mathsf{ABE}$ formula $\varphi_\HS$ over $\mathpzc{AP}$ such that for all
  Kripke structures $\mathpzc{K}$ over $\mathpzc{AP}$ and traces $\rho$, $\mathpzc{K},\rho,0\models \varphi$ if and only if  $\mathpzc{K},\rho\models_\stat \varphi_\HS$.
  \end{theorem}
  \begin{proof} The proof is by induction on the nesting depth of modality $\EQF$ in $\varphi$. In the base case, $\varphi$ is a finitary $\LTL$ formula over $\mathpzc{AP}$. Since what we need to deal with it is just the first part of the work we have to do for the inductive step, it is omitted and only the inductive step is detailed. 
  
  Let $H$ be the non-empty set of subformulas of $\varphi$ of the form $\EQF \psi$ which do not occur in the scope of the path quantifier $\EQF$, that is, the $\EQF \psi$ formulas which are maximal with respect to the nesting depth of modality $\EQF$. Then, $\varphi$ can be seen as an $\LTL$ formula over the extended set of proposition letters
  $\overline{\mathpzc{AP}}=\mathpzc{AP}\cup H$. Let $\Sigma = 2^{\overline{\mathpzc{AP}}}$ and $\overline{\varphi}$ be the $\LTL$ formula over $\Sigma$
  obtained from $\varphi$ by replacing the occurrences of each proposition letter $p \in \overline{\mathpzc{AP}}$ in $\varphi$ with the formula
  $\bigvee_{P\in \Sigma \; : \; p\in P} P$, according to the $\LTL$ action-based semantics. 
  
  Given a Kripke structure  $\mathpzc{K}$ over $\mathpzc{AP}$ with labeling $\mu$  and a trace $\rho$ of $\mathpzc{K}$,
  we denote by $\rho_H$ the finite word over $2^{\overline{\mathpzc{AP}}}$ of length $|\rho|$ defined as $\rho_H(i)= \mu(\rho(i))\cup \{\EQF\psi\in H\mid\mathpzc{K},\rho,i\models \EQF\psi\}$, for all $i\in [0,|\rho|-1]$.
  One can easily prove by structural induction on $\overline{\varphi}$ that $\mathpzc{K},\rho,0\models \varphi$ if and only if $\rho_H\in \Lang_\act(\overline{\varphi})$.
  By Theorem~\ref{theoremFromLTLtoBEoverFiniteWords}, there exists a $\mathsf{BE}$ formula $\overline{\varphi}_\HS$ over $\Sigma$ such that
  $\Lang_\act(\overline{\varphi})= \Lang_\act(\overline{\varphi}_\HS)$.

Now, by the induction hypothesis, for each formula   $\EQF \psi\in H$, there exists an $\mathsf{ABE}$ formula $\psi_\HS$ such that
for all Kripke structures $\mathpzc{K}$ and traces $\rho$ of $\mathpzc{K}$,
$\mathpzc{K},\rho,0\models \psi \text{ iff }  \mathpzc{K},\rho\models_\stat \psi_\HS$.
Since $\rho$ is arbitrary, 
$\mathpzc{K},\rho,i\models \EQF\psi\, \text{ iff }\, \mathpzc{K},\rho[i,i],0\models \EQF\psi\, \text{ iff }\,  \mathpzc{K},\rho[i,i]\models_\stat \hsA\psi_\HS$,
for each $i\geq 0$.

Let $\varphi_\HS$ be the $\mathsf{ABE}$ formula over $\mathpzc{AP}$ obtained from the $\mathsf{BE}$ formula $\overline{\varphi}_\HS$ by replacing each occurrence of
$P\in \Sigma$ in $\overline{\varphi}_\HS$ with the formula
\[
[G]\Big(\Length_1 \longrightarrow  \smashoperator[r]{\bigwedge_{\EQF\psi\in H\cap P}} \hsA\psi_\HS \;\wedge  \smashoperator[r]{\bigwedge_{\EQF\psi\in H\setminus P}}\neg \hsA\psi_\HS \;\wedge \smashoperator[r]{\bigwedge_{p\in \mathpzc{AP}\cap P}} p\;\wedge  \smashoperator[r]{\bigwedge_{p\in \mathpzc{AP}\setminus P}} \neg p\Big).
\]

Since for all $i\geq 0$ and $\EQF\psi\in H$, $\mathpzc{K},\rho,i\models \EQF\psi$ if and only if  $\mathpzc{K},\rho[i,i]\models_\stat \hsA\psi_\HS$, it is possible to prove by a straightforward induction
on the structure of $\overline{\varphi}_\HS$ that, for any Kripke structure $\mathpzc{K}$ and trace $\rho$ of $\mathpzc{K}$ we have
$\mathpzc{K},\rho\models_\stat \varphi_\HS$ if and only if $\rho_H\in \Lang_\act(\overline{\varphi}_\HS)$.
%

   Therefore, since $\mathpzc{K},\rho,0\models \varphi$ if and only if $\rho_H\in \Lang_\act(\overline{\varphi})$ and
   $\Lang_\act(\overline{\varphi})= \Lang_\act(\overline{\varphi}_\HS)$, $\mathpzc{K},\rho,0\models \varphi$ if and only if  $\mathpzc{K},\rho\models_\stat \varphi_\HS$, for any Kripke structure
 $\mathpzc{K}$ and trace $\rho$ of $\mathpzc{K}$. 
  \end{proof}

  Since the fragment $\A\B\E$ of $\HS$ does not feature any modalities unravelling a Kripke structure backward (namely, $\hsAt$ and $\hsEt$), the computation-tree-based semantics coincides with the state-based one (recall Figure~\ref{fig:ST} and \ref{fig:CT}), and thus the next corollary immediately follows from Theorem~\ref{theo:FromFInitaryCTLStarToFutureHSBranching}. 

  \begin{corollary}\label{cor:FromFInitaryCTLStarToFutureHSBranching} Finitary $\CTLStar$ is subsumed by both $\HS_\stat$ and $\HS_\LinearPast$.
  \end{corollary}

\subsection{From $\HS_\LinearPast$ to finitary $\CTLStar$}
We show now that $\HS_\LinearPast$ is subsumed by both $\CTLStar$ and its finitary variant.
To prove this result,  we first introduce a hybrid and linear-past extension of
$\CTLStar$, called \emph{hybrid} $\CTLStarLP$, and its finitary variant, called \emph{finitary hybrid} $\CTLStarLP$.

Besides standard modalities, hybrid logics make use of explicit variables and quantifiers that bind them~\cite{BS98}.
Variables and binders allow us to easily mark points in a path, which will be considered as starting and ending points of intervals, thus permitting a natural encoding of $\HS_\LinearPast$. Actually, we will show that the restricted use of variables and binders exploited in our encoding does not increase the expressive power of (finitary) $\CTLStar$ (as it happens for an unrestricted use), thus proving the desired result. We start defining \emph{hybrid} $\CTLStarLP$.
%

For a countable set $\{x,y,z,\ldots\}$ of (position) variables, the set of  formulas $\varphi$ of
hybrid $\CTLStarLP$  over $\mathpzc{AP}$ is defined as follows:
\[
\varphi ::= \top \ | \ p \ | \ x \ | \ \neg \varphi \ | \ \varphi \vee \varphi \ | \ \Downarrowx.\varphi \ | \ \Next \varphi\ | \ \varphi \until \varphi \ | \ \Next^{-} \varphi\ | \ \varphi \until^{-} \varphi\ | \ \EQ  \varphi,
\]
where $\Next^{-}$ (`previous') and $\until^{-}$ (`since') are the  past counterparts of the `next' and `until' modalities $\Next$ and $\until$, and $\Downarrowx$ is the \emph{downarrow binder operator}~\cite{BS98}, which binds $x$ to the current position along the given initial infinite path. We also use the standard shorthands $\Eventually^{-}\varphi:= \top \until^{-} \varphi$ (`eventually in the past')  and its
dual  $\Always^{-} \varphi:=\neg \Eventually^{-}\neg\varphi$ (`always in the past'). As usual, a sentence is a formula with no free variables.

Let $\mathpzc{K}$ be a Kripke structure and $\varphi$ be a hybrid $\CTLStarLP$ formula. For an \emph{initial}  infinite path $\pi$ of $\mathpzc{K}$, a variable valuation $g$, that assigns to each  variable $x$  a position along $\pi$, and  $i\geq 0$, the satisfaction relation 
$\pi,g, i\models  \varphi$ 
is defined as follows (we omit the clauses for Boolean connectives, for $\until$ and $\Next$):
\[ \begin{array}{ll}
\pi, g,i \models \Next^{-} \varphi  & \Leftrightarrow  i>0 \text{ and } \pi, g, i-1 \models \varphi ,\\
 \pi, g, i \models \varphi_1\until^{-} \varphi_2  &
  \Leftrightarrow  \text{for some $j\leq i$}, \pi, g, j
  \models \varphi_2
  \text{ and }  \pi, g, k \models  \varphi_1 \text{ for all }j< k\leq i,\\
\pi, g, i \models \EQ \varphi  & \Leftrightarrow \text{for some initial infinite path } \pi'  \text{ such that } \pi'[0,i]=\pi[0,i],\,   \pi', g, i \models \varphi,\\
\pi,g, i \models x  &  \Leftrightarrow  g(x)=i,\\
\pi,g, i \models \Downarrowx.\varphi  &  \Leftrightarrow  \pi,g[x\leftarrow i], i \models \varphi,
\end{array} \]
where $g[x\leftarrow i](x)=i$ and $g[x\leftarrow i](y)=g(y)$ for $y\neq x$.  A Kripke structure $\mathpzc{K}$ is a model of a formula $\varphi$ if  $\pi,g_0, 0 \models \varphi$, for every initial infinite path $\pi$ of $\mathpzc{K}$, with $g_0$ the variable evaluation assigning $0$ to each variable. Note that the path quantification is `memoryful',
i.e., it  ranges over infinite paths that start at the root and visit  the current node of the computation tree. Clearly, the semantics for the syntactical fragment
$\CTLStar$ coincides with the standard one. If we disallow the use of variables and binder modalities, we obtain the  logic
$\CTLStarLP$, a well-known linear-past extension of $\CTLStar$ which is as expressive as $\CTLStar$ \cite{jcss/KupfermanPV12}.
We also consider the finitary variant of hybrid $\CTLStarLP$, where the path quantifier $\EQ$ is replaced with the finitary path quantifier $\EQF$. This logic corresponds to an extension of finitary $\CTLStar$ and its semantics is similar to that of hybrid $\CTLStarLP$ with the exception that path quantification ranges over the \emph{finite} paths (traces) that start at the root and visit the current node of the computation tree.

In the following, we will use the fragment of hybrid $\CTLStarLP$ consisting of \emph{well-formed} formulas,
namely, formulas $\varphi$ where:
 \begin{itemize}
   \item each subformula $\EQ \psi$   of $\varphi$ has at most one free variable (namely, not bound by the downarrow binder operator); 
	\item each  subformula $\EQ \psi(x)$ of $\varphi$ having $x$ as free variable occurs in $\varphi$ in the context $(\Eventually^{-} x) \wedge \EQ\psi(x)$.
\end{itemize}
Intuitively, the above conditions affirm that, for each state subformula  $\EQ \psi$, the unique free variable (if any) refers to ancestors of the current node in the computation tree.\footnote{The well-formedness constraint ensures that a formula captures only branching regular requirements. As an example,
the formula $\EQ\Eventually\Downarrowx.\Always^{-}(\neg \Next^{-} \top \rightarrow \AQ \Eventually(x\wedge p))$ is \emph{not} well-formed and requires that there is a level of the computation tree such that each node in the level satisfies $p$. This represents a non-regular context-free branching requirement (see, e.g., \cite{AlurCZ06}).} 

The notion of well-formed formula of finitary hybrid $\CTLStarLP$ is similar: the path quantifier $\EQ$ is replaced by its finitary version $\EQF$.

 We first show that
 $\HS_\LinearPast$ can be translated into the well-formed fragment of hybrid $\CTLStarLP$ (resp., well-formed fragment of finitary hybrid $\CTLStarLP$). Then, we show that this fragment is subsumed by $\CTLStar$ (resp., finitary $\CTLStar$).

\begin{proposition}\label{prop:HSComputationToHybrid} Given a $\HS_\LinearPast$ formula $\varphi$, one can construct in linear-time an equivalent well-formed
sentence of hybrid $\CTLStarLP$ (resp., finitary hybrid $\CTLStarLP$).
\end{proposition}
\begin{proof}
We focus on the translation from $\HS_\LinearPast$ into the well-formed fragment of hybrid $\CTLStarLP$. The translation from $\HS_\LinearPast$ into the well-formed fragment
of finitary hybrid $\CTLStarLP$ is similar, and thus omitted.
Let $\varphi$ be a $\HS_\LinearPast$ formula. The desired hybrid $\CTLStarLP$ sentence is the formula
$\Downarrowx .\Always\, f(\varphi,x)$, where $f(\varphi,x)$ is a mapping which is homomorphic with respect to 
the Boolean connectives, and over proposition letters and modalities behaves as follows:
\[ \begin{array}{ll}
f(p,x) & = \Always^{-}((\Eventually^{-}x ) \rightarrow p),\\
f(\hsB\psi,x) & = \Next^{-}\Eventually^{-}(f(\psi,x)\wedge \Eventually^{-}x) ,\\
f(\hsBt\psi,x) & = \EQ(\Next\Eventually f(\psi,x)) \wedge \Eventually^{-}x,\\
f(\hsE\psi,x) & = \Downarrowy.\Eventually^{-}\bigl(x \wedge \Next\Eventually \Downarrowx.\Eventually(y\wedge f(\psi,x))\bigr ) ,\\
f(\hsEt\psi,x) & = \Downarrowy.\Eventually^{-}\bigl((\Next\Eventually x) \wedge  \Downarrowx.\Eventually(y\wedge f(\psi,x))\bigr ),
\end{array} \]
where $y$ is a fresh variable.

Clearly $\Downarrowx .\Always\, f(\varphi,x)$ is well-formed. The formula $f(\varphi,x)$ intuitively states that $\varphi$ holds over an 
interval of the current path that starts at the position (associated with the variable) $x$ and ends at the current position.
More formally, let $\mathpzc{K}$ be a Kripke structure, $[h,i]$ be an interval of positions, $g$ be a valuation assigning to 
the variable $x$ the position $h$, and $\pi$ be an initial infinite path.
By a straightforward induction on the structure of  $\varphi$, one can show that
 $\mathpzc{K},\pi,g,i\models f(\varphi,x)$ if and only if  $\mathpzc{C}(\mathpzc{K}),\mathpzc{C}(\pi,h,i)\models_\stat \varphi$, where
$\mathpzc{C}(\pi,h,i)$ denotes the trace of the computation tree $\mathpzc{C}(\mathpzc{K})$ starting from $\pi[0,h]$ and leading to $\pi[0,i]$.
Hence, $\mathpzc{K}$ is a model of $\Downarrowx .\Always\, f(\varphi,x)$ if, for each initial trace $\rho$ of $\mathpzc{C}(\mathpzc{K})$, we have
$\mathpzc{C}(\mathpzc{K}),\rho\models_\stat \varphi$.
\end{proof}

Let $\LTLP$ be the past extension of $\LTL$, obtained by adding the past modalities $\Next^{-}$ and $\until^{-}$.
By exploiting the well-known separation theorem for $\LTLP$ over finite and infinite words \cite{Gabbay87}, which states that any $\LTLP$ formula can be effectively converted into an equivalent Boolean combination of $\LTL$ formulas and pure past $\LTLP$ formulas, we can prove that, under the hypothesis of well-formedness, the extensions of $\CTLStar$ (resp., finitary $\CTLStar$) used to encode  $\HS_\LinearPast$ formulas do not increase the expressive power of $\CTLStar$ (resp., finitary $\CTLStar$). Such a result is the fundamental step to prove, together with Proposition~\ref{prop:HSComputationToHybrid},
that $\CTLStar$ subsumes $\HS_\LinearPast$. In addition, paired with Corollary~\ref{cor:FromFInitaryCTLStarToFutureHSBranching}, it will allow us to state the main result of the section, namely, that $\HS_\LinearPast$ and finitary $\CTLStar$ have the same expressiveness.

Let us now show that the well-formed fragment of hybrid $\CTLStarLP$ (resp., finitary hybrid $\CTLStarLP$) is not more expressive than $\CTLStar$ (resp., finitary $\CTLStar$). Once more, we focus on the well-formed fragment of hybrid $\CTLStarLP$ omitting the similar proof for the finitary variant.

We start with some additional definitions and auxiliary results.  A \emph{pure past} $\LTLP$ formula is an $\LTLP$ formula which does not contain occurrences of future temporal modalities. Given two formulas $\varphi$ and $\varphi'$ of hybrid $\CTLStarLP$, we say that
$\varphi$ and $\varphi'$ are \emph{congruent} if, for every  Kripke structure $\mathpzc{K}$, initial infinite path $\pi$, valuation $g$, and current position $i$,
$\mathpzc{K},\pi,g, i \models \varphi$ if and only if $\mathpzc{K},\pi,g, i \models \varphi'$ (note that congruence is a \emph{stronger} requirement than equivalence).

As usual, for a formula $\varphi$ of hybrid $\CTLStarLP$ with one free variable $x$, we write $\varphi(x)$. Moreover, since the satisfaction relation depends only on the variables occurring free in the given formula, for $\varphi(x)$ we  use the notation
$\mathpzc{K},\pi,i \models \varphi(x \leftarrow h)$ to mean that $\mathpzc{K},\pi,g,i \models \varphi$ for any valuation $g$ assigning $h$ to the unique free variable $x$.
For a formula $\varphi$ of hybrid $\CTLStarLP$, let $\EQSubf(\varphi)$ denote the set of subformulas of $\varphi$ of the form $\EQ \psi$ which do not occur in the scope of
the path quantifier $\EQ$.

Finally, for technical reasons, we introduce the notion of \emph{simple} hybrid $\CTLStarLP$ formula.
\begin{definition} Given a variable $x$, a \emph{simple} hybrid $\CTLStarLP$ formula $\psi$ with respect to $x$  is a hybrid $\CTLStarLP$ formula satisfying the following syntactical constraints:
\begin{itemize}
    \item $x$ is the unique variable occurring in $\psi$; 
    \item $\psi$ does \emph{not} contain occurrences of the binder modalities and past temporal modalities;
    \item $\EQSubf(\psi)$ consists of $\CTLStar$ formulas.
\end{itemize}
\end{definition}

Intuitively, a \emph{simple} hybrid $\CTLStarLP$ (over $\mathpzc{AP}$) formula $\psi$ with respect to $x$  can be seen as a $\CTLStar$ formula over the set of proposition letters $\mathpzc{AP}\cup \{x\}$ such that $x$ does not occur in the scope of $\EQ$. The next lemma shows that $\psi$ can be further simplified whenever it is paired with the formula $\Eventually^{-} x$.

\begin{lemma}\label{lemma:UsingSeparationHybridCTLPreliminary} Let $\psi$ be  a \emph{simple} hybrid $\CTLStarLP$ formula  with respect to $x$. 
Then, $(\Eventually^{-} x)\wedge \psi$ is congruent to a formula of the form $(\Eventually^{-} x)\wedge \xi$, where $\xi$ is a Boolean combination of the atomic formula $x$ and $\CTLStar$ formulas. 
\end{lemma}
\begin{proof} Let  $\psi$ be  a \emph{simple} hybrid $\CTLStarLP$ formula  with respect to $x$. From a syntactic point of view, $\psi$ is not, in general, a $\CTLStar$ formula due to the occurrences of the free variable $x$. We show that these occurrences can be separated whenever $\psi$ is paired with $\Eventually^{-} x$, obtaining a Boolean combination of the atomic formula $x$ and $\CTLStar$ formulas. 

The base case with $\psi=x$, $\psi=p\in\mathpzc{AP}$, or $\psi=\EQ\psi'$ is obvious.

As for the inductive step, let $\psi$ be a Boolean combination of simple hybrid $\CTLStarLP$ formulas $\theta$, where $\theta$ is either $p\in\mathpzc{AP}$, the variable $x$, a $\CTLStar$ formula, or a simple hybrid $\CTLStarLP$ formula (with respect to $x$) of the forms $\Next\theta_1$ or 
$\theta_1\until \theta_2$. Therefore, we just need to consider the cases where $\theta =\Next \theta_1$ or $\theta =\theta_1\until \theta_2$.

Let us consider the case $\theta =\Next \theta_1$. Since there are not past temporal modalities in $\theta_1$, 
$\Next \theta_1$ forces the free occurrence of $x$ in $\psi$ to be interpreted in a (strictly) future position. 
However, $\psi$ is conjunct with the formula $\Eventually^{-} x$, which turns out to be false when $x$ is associated with a (strictly) future position.  
Let us denote by $\widehat{\theta}$  the $\CTLStar$ formula obtained from $\theta$ by replacing each occurrence of $x$ in $\psi$ with $\bot$ (false).
Now, when $x$ is mapped to a (strictly) future position, $\Eventually^{-} x$ is false, and, when $x$ is mapped to a present/past position, $\Eventually^{-} x$ is  true, and $\theta$ and $\widehat{\theta}$ are congruent.
As a consequence, it is clear that $(\Eventually^{-} x)\wedge \theta$ is congruent to $(\Eventually^{-} x)\wedge  \widehat{\theta}$.

Let us consider the case for $\theta =\theta_1\until \theta_2$. Using the same arguments of the previous case, we have that
 $(\Eventually^{-} x)\wedge \theta$ is congruent to $(\Eventually^{-} x)\wedge (\theta_2 \vee (\theta_1\wedge\Next (\widehat{\theta_1\until \theta_2})))$. By distributivity of $\wedge$ over $\vee$, we get $((\Eventually^{-} x)\wedge\theta_2) \vee ((\Eventually^{-} x)\wedge\theta_1\wedge\Next (\widehat{\theta_1\until \theta_2})))$. The thesis follows by applying the inductive hypothesis to $(\Eventually^{-} x)\wedge\theta_2$ and to $(\Eventually^{-} x)\wedge\theta_1$, and by factorizing $\Eventually^{-} x$ (notice that $\widehat{\theta_1\until \theta_2}$ is a $\CTLStar$ formula).
\end{proof}

The next lemma states an important technical property of  well formed formulas, which will be exploited in 
Theorem~\ref{Theo:HybridCTLInfiniteCase} to prove that the set of sentences of the well-formed fragment of hybrid $\CTLStarLP$ has the same expressiveness as $\CTLStar$. Intuitively, if the hybrid features of the language do not occur in the scope of existential path quantifiers, it is possible to remove the occurrences of the binder $\downarrow$ and to suitably separate past and future modalities. The result is obtained by exploiting the equivalence of $\FO$ and $\LTLP$ over infinite words and by applying the separation theorem for $\LTLP$ over infinite words~\cite{Gabbay87}, that we recall here for completeness. 

\begin{theorem}[$\LTLP$ separation over infinite words]\label{th:LTLpsep}
Any $\LTLP$ formula $\psi$ can be effectively transformed into a formula
\[
\psi'=   \bigvee_{i=1}^t(\psi_{p,i}\wedge \psi_{f,i}),
\]
for some $t\geq 1$, where $\psi_{p,i}$ is a pure past $\LTLP$ formula and $\psi_{f,i}$ is an $\LTL$ formula, such that
for all infinite words $w$ over $2^{\mathpzc{AP}}$ and $i\geq 0$, it holds that
$
 w,i\models \psi \text{ if and only if } w,i\models \psi'.
$
\end{theorem}

\begin{lemma}\label{lemma:UsingSeparationHybridCTL} Let $(\Eventually^{-} x)\wedge \EQ \varphi(x)$ (resp., $\EQ\varphi$) be a well-formed formula (resp., well-formed sentence) of hybrid $\CTLStarLP$ such that $\EQSubf(\varphi)$ consists of
$\CTLStar$ formulas. Then, $(\Eventually^{-} x)\wedge \EQ \varphi(x)$ (resp., $\EQ\varphi$) is congruent to a well-formed formula of hybrid $\CTLStarLP$ which is a Boolean combination of $\CTLStar$ formulas and (formulas that
correspond to) \emph{pure past}  $\LTLP$ formulas over the set of proposition letters $\mathpzc{AP}\cup  \EQSubf(\varphi) \cup \{x\}$ (resp., $\mathpzc{AP}\cup  \EQSubf(\varphi)$).
\end{lemma}
\begin{proof} We focus on well-formed formulas of the form  $(\Eventually^{-} x)\wedge \EQ \varphi(x)$. The case of well-formed sentences of the form $\EQ\varphi$ is similar, and thus omitted.

Let $\overline{\mathpzc{AP}}= \mathpzc{AP}\cup  \EQSubf(\varphi) \cup \{x\}$. By hypothesis, $\EQSubf(\varphi)$ is a set of $\CTLStar$ formulas, that is, they are devoid of any hybrid feature.

Given a Kripke structure $\mathpzc{K}=(\mathpzc{AP},S, \delta,\mu,s_0)$,  an initial infinite path $\pi$, and $h\geq 0$, we denote by $\pi_{\overline{\mathpzc{AP}},h}$ the infinite word over $2^{\overline{\mathpzc{AP}}}$, which, for every position $i\geq 0$, is defined as follows:
 \begin{itemize}
 \item $\pi_{\overline{\mathpzc{AP}},h}(i)\cap \mathpzc{AP} =\mu(\pi(i))$;
 \item $\pi_{\overline{\mathpzc{AP}},h}(i)\cap \EQSubf(\varphi)= \{\psi\in \EQSubf(\varphi)\mid \mathpzc{K},\pi,i\models \psi\}$;
 \item $x\in \pi_{\overline{\mathpzc{AP}},h}(i)$ if and only if $i=h$.
 \end{itemize}

 By using a fresh position variable $\present$ to represent the current position, the formula $\varphi(x)$ can be easily converted into an $\FO$ formula
 $\varphi_\FO(\present)$
 over $ \overline{\mathpzc{AP}}$ having $\present$ as its unique free variable, such that for all Kripke structures $\mathpzc{K}$, initial infinite paths $\pi$, and positions $i$ and $h$, we have:
 \begin{equation}\label{eq:UsingSeparationHybridCTL1}
  \mathpzc{K},\pi,i\models \varphi(x \leftarrow h) \text{ if and only if } \pi_{\overline{\mathpzc{AP}},h} \models \varphi_\FO(\present \leftarrow i).
\end{equation}
(To this end, it suffices to map any proposition letter $\overline{p}\in \overline{\mathpzc{AP}}$ into a unary predicate $\overline{p}$,  and all the operators $\Next,\Next^{-},\until,\until^{-},\downarrow$ into  $\FO$ formulas expressing their semantics.)

 By the equivalence of $\FO$ and $\LTLP$ and the separation theorem for $\LTLP$ over infinite words (Theorem~\ref{th:LTLpsep}), starting from the $\FO$ formula $\varphi_\FO(\present)$, one can construct an $\LTLP$ formula $\varphi_\LTLP$ over $\overline{\mathpzc{AP}}$ of the form 
  \begin{equation}\label{eq:UsingSeparationHybridCTL2}
\varphi_\LTLP:=   \bigvee_{i\in I}(\varphi_{p,i}\wedge \varphi_{f,i})
\end{equation}
such that $\varphi_{p,i}$ 
is a pure past $\LTLP$ formula,
$\varphi_{f,i}$ 
is an $\LTL$ formula, and for all infinite words $w$ over $2^{\overline{\mathpzc{AP}}}$ and $i\geq 0$, it holds that:
  \begin{equation}\label{eq:UsingSeparationHybridCTL3}
 w,i\models \varphi_\LTLP \text{ if and only if } w\models \varphi_\FO(\present \leftarrow i).
\end{equation}
 The $\LTLP$ formula $\varphi_\LTLP$ over $\overline{\mathpzc{AP}}$ corresponds to a hybrid $\CTLStarLP$ formula $\varphi_\LTLP(x)$ over $\mathpzc{AP}$. (Note that the only hybrid feature is the possible occurrence of the variable $x$.)
 By definition of the infinite words $\pi_{\overline{\mathpzc{AP}},h}$, one can easily show by structural induction that 
  for all Kripke structures $\mathpzc{K}$, initial infinite paths $\pi$, and positions $i$ and $h$:
  \begin{equation}\label{eq:UsingSeparationHybridCTL4}
\pi_{\overline{\mathpzc{AP}},h},i \models  \varphi_\LTLP \text{ if and only if  }  \mathpzc{K},\pi,i\models \varphi_\LTLP(x \leftarrow h),
\end{equation}
the latter being a hybrid $\CTLStarLP$ formula.
Thus, by Points (\ref{eq:UsingSeparationHybridCTL1}), (\ref{eq:UsingSeparationHybridCTL3}), and (\ref{eq:UsingSeparationHybridCTL4}), we obtain that 
$\varphi(x)$ and $\varphi_\LTLP(x)$ are congruent.

Since in (\ref{eq:UsingSeparationHybridCTL2}),  for each $i\in I$, $\varphi_{p,i}$ is a pure past $\LTLP$ formula over $\overline{\mathpzc{AP}}$,  $\EQ \varphi_{p,i}(x)$ is trivially congruent to $\varphi_{p,i}(x)$. As a consequence, we have that
$(\Eventually^{-} x)\wedge \EQ \varphi(x)$ is congruent to
 $(\Eventually^{-} x)\wedge \bigvee_{i\in I}(\varphi_{p,i}(x)\wedge \EQ\varphi_{f,i}(x))$, which is congruent to $\bigvee_{i\in I}(\varphi_{p,i}(x)\wedge (\Eventually^{-} x)\wedge \EQ\varphi_{f,i}(x))$, which is in turn congruent to
 $\bigvee_{i\in I}(\varphi_{p,i}(x)\wedge \EQ((\Eventually^{-} x)\wedge \varphi_{f,i}(x)))$.

 Now, $\varphi_{f,i}(x)$ is a \emph{simple} hybrid $\CTLStarLP$ formula  with respect to $x$, and  $\EQ x$ (resp., $\EQ\neg x$) is trivially congruent to $x$ (resp., $\neg x$).
 By Lemma~\ref{lemma:UsingSeparationHybridCTLPreliminary} and some simple manipulation steps, we can prove the following sequence of equivalences:
 
 \begin{align*}
 \bigvee_{i\in I}\Big(\varphi_{p,i}(x)\wedge \EQ((\Eventually^{-} x)\wedge \varphi_{f,i}(x))\Big)&=
 \tag*{(Lemma~\ref{lemma:UsingSeparationHybridCTLPreliminary} and disjunctive normal form)}\\
   \bigvee_{i\in I}\Big(\varphi_{p,i}(x) \wedge\EQ \big((\Eventually^{-} x)\wedge\bigvee_{j\in J}(\tilde{x}_{i,j}\wedge \psi_{i,j})\big)\Big)& =\tag*{($\Eventually^{-} x$ is a pure past $\LTLP$ formula)}\\
   \bigvee_{i\in I}\Big(\varphi_{p,i}(x) \wedge(\Eventually^{-} x)\wedge\EQ \bigvee_{j\in J}(\tilde{x}_{i,j}\wedge \psi_{i,j})\Big)&=\tag*{(Distributive property of $\wedge$ over $\vee$)}\\
 (\Eventually^{-} x)\wedge \bigvee_{i\in I}\Big(\varphi_{p,i}(x) \wedge\EQ \bigvee_{j\in J}(\tilde{x}_{i,j}\wedge \psi_{i,j})\Big)&=\tag*{(Distributive property of $\EQ$ over $\vee$ and $\tilde{x}_{i,j}$ is a pure past $\LTLP$ formula)}\\
 (\Eventually^{-} x)\wedge \bigvee_{i\in I}\Big(\varphi_{p,i}(x) \wedge \bigvee_{j\in J}(\tilde{x}_{i,j}\wedge \EQ\psi_{i,j})\Big)&=\tag*{(Distributive property of $\wedge$ over $\vee$)}\\
 (\Eventually^{-} x)\wedge \bigvee_{i\in I}\bigvee_{j\in J}\Big(\varphi_{p,i}(x) \wedge \tilde{x}_{i,j}\wedge \EQ\psi_{i,j}\Big)&
 \end{align*}
where $\tilde{x}_{i,j}$ is either $x$, $\neg x$, or $\top$.
 
 Hence, 
 $(\Eventually^{-} x)\wedge \EQ \varphi(x)$ is congruent to a formula of the form
 $(\Eventually^{-} x)\wedge \bigvee_{i\in I'}(\psi_{p,i}(x)\wedge \EQ\psi_{i})$, for some $I'$, where $\psi_{p,i}(x)$ corresponds 
 to a  \emph{pure past}  $\LTLP$ formula over $\overline{\mathpzc{AP}} \,(=\mathpzc{AP}\cup  \EQSubf(\varphi) \cup \{x\})$
 and
 $\psi_i$ is a $\CTLStar$ formula.
\end{proof}

The following lemma generalizes the separation result given by Lemma~\ref{lemma:UsingSeparationHybridCTL} to any well-formed formula of the form $(\Eventually^{-} x)\wedge \EQ \varphi(x)$, that is, to formulas where $\varphi(x)$ is unconstrained.

\begin{lemma}\label{cor:UsingSeparationHybridCTL} Let $(\Eventually^{-} x)\wedge \EQ \varphi(x)$ (resp., $\EQ\varphi$) be a well-formed formula (resp., well-formed sentence) of hybrid $\CTLStarLP$.  Then, there exists a finite set $\mathpzc{H}$ of $\CTLStar$ formulas of the form $\EQ\psi$, such that $(\Eventually^{-} x)\wedge \EQ \varphi(x)$ (resp., $\EQ\varphi$) is congruent to a well-formed formula of hybrid $\CTLStarLP$ which is a Boolean combination of $\CTLStar$ formulas and (formulas that
correspond to) \emph{pure past}  $\LTLP$ formulas over the set of proposition letters $\mathpzc{AP}\cup  \mathpzc{H} \cup \{x\}$ (resp., $\mathpzc{AP}\cup  \mathpzc{H}$).
\end{lemma} 
\begin{proof}
As in the case of Lemma~\ref{lemma:UsingSeparationHybridCTL}, we focus on well-formed formulas of the form  $(\Eventually^{-} x)\wedge \EQ \varphi(x)$ (the case of well-formed sentences of the form $\EQ\varphi$ is similar).

The proof is by induction on the nesting depth of the path quantifier $\exists$ in $\varphi(x)$. 

Base case: $\EQSubf(\varphi)=\emptyset$. We apply Lemma~\ref{lemma:UsingSeparationHybridCTL}, and the result follows taking  $\mathpzc{H}=\emptyset$. 

Inductive step: 
let $\EQ \psi\in \EQSubf(\varphi)$. Since $(\Eventually^{-} x)\wedge \EQ \varphi(x)$ is well-formed, either $\psi$ is a sentence, or $\psi$ has a unique free variable $y$ and $\EQ \psi(y)$ occurs in $\varphi(x)$ in the context $(\Eventually^{-} y)\wedge \EQ \psi(y)$. Assume that the latter case holds (the former is similar). 
By definition of well-formed formula, $y$ is not free in $\varphi(x)$, and $(\Eventually^{-} y)\wedge \EQ \psi(y)$ must occur in the scope of some occurrence of  $\Downarrowy$.
By the inductive hypothesis, the thesis holds for $(\Eventually^{-} y)\wedge \EQ \psi(y)$. Hence, there exists 
a finite set $\mathpzc{H}'$ of $\CTLStar$ formulas of the form $\EQ\theta$ such that $(\Eventually^{-} y)\wedge \EQ \psi(y)$ is congruent to a well-formed formula of hybrid $\CTLStarLP$, say $\xi(y)$, which is a Boolean combination of $\CTLStar$ formulas and formulas that correspond to  \emph{pure past}  $\LTLP$ formulas over the set of proposition letters
$\mathpzc{AP}\cup  \mathpzc{H'} \cup \{y\}$. 

By replacing each occurrence of $(\Eventually^{-} y)\wedge \EQ \psi(y)$ in $\varphi(x)$ with $\xi(y)$, and repeating the procedure for all the formulas in   $\EQSubf(\varphi)$, we obtain a well-formed formula of hybrid $\CTLStarLP$ of the form $(\Eventually^{-} x)\wedge \EQ \theta(x)$ which is congruent to $(\Eventually^{-} x)\wedge \EQ \varphi(x)$ (note that the congruence relation is closed under substitution) and such that $\EQSubf(\theta)$ consists of $\CTLStar$ formulas. At this point we can apply Lemma~\ref{lemma:UsingSeparationHybridCTL} proving the assertion.
\end{proof}

We can now prove that the well-formed sentences of hybrid $\CTLStarLP$ can be expressed in $\CTLStar$.

\begin{theorem}\label{Theo:HybridCTLInfiniteCase} 
The set of sentences of the well-formed fragment of hybrid $\CTLStarLP$ has the same expressiveness as $\CTLStar$.
\end{theorem}
\begin{proof}
Let $\varphi$ be a well-formed sentence of hybrid $\CTLStarLP$. To prove the thesis,  we construct a $\CTLStar$ formula which is equivalent to $\varphi$. 

Since $\varphi$ is equivalent to $\neg \EQ \neg \varphi$ and $\neg \EQ \neg \varphi$ is well-formed, by applying Lemma~\ref{cor:UsingSeparationHybridCTL} one can convert  $\neg \EQ \neg \varphi$ into a congruent hybrid $\CTLStarLP$ formula which is a Boolean combination of $\CTLStar$ formulas and formulas $\theta$ which can be seen as pure past $\LTLP$ formulas over the set of proposition letters $\mathpzc{AP}\cup  \mathpzc{H}$, where $\mathpzc{H}$ is a set of $\CTLStar$ formulas of the form $\EQ\psi$. 

Since the past temporal modalities
in such  $\LTLP$ formulas $\theta$ refer to the initial position of the initial infinite paths, one can replace $\theta$ with an equivalent $\CTLStar$ formula $f(\theta)$, where the mapping $f$ is inductively defined as follows:
 \begin{itemize}
   \item $f(p)=p$ for all $p\in \mathpzc{AP}\cup  \mathpzc{H}$;
   \item $f$ is homomorphic with respect to the Boolean connectives;
   \item $f(\Next^{-}\theta)=\bot$ and $f(\theta_1\until^{-}\theta_2)=f(\theta_2)$.
 \end{itemize}
The resulting $\CTLStar$ formula is equivalent to $\neg \EQ \neg \varphi$, as required.
\end{proof}

By an easy adaptation of the proof of Theorem~\ref{Theo:HybridCTLInfiniteCase}, where one exploits the separation theorem for $\LTLP$ over finite words \cite{Gabbay87}, it is possible to characterize also the expressiveness of well-formed \emph{finitary} hybrid $\CTLStarLP$.

\begin{theorem}\label{theo:FinitaryHybridCTL} The set of sentences of the well-formed fragment of finitary hybrid $\CTLStarLP$ has the same expressiveness as finitary $\CTLStar$.
\end{theorem}


Together with Proposition~\ref{prop:HSComputationToHybrid}, Theorem~\ref{Theo:HybridCTLInfiniteCase} 
(resp., Theorem~\ref{theo:FinitaryHybridCTL})
allows us to conclude that $\CTLStar$ (resp., finitary $\CTLStar$) subsumes $\HS_\LinearPast$. 

Finally, by exploiting Corollary~\ref{cor:FromFInitaryCTLStarToFutureHSBranching}, we can state the main result of the section, namely, $\HS_\LinearPast$ and finitary $\CTLStar$ have the same expressiveness.

\begin{theorem}\label{Cor:CharacterizationHSCompTree} $\CTLStar\geq\HS_\LinearPast$. Moreover, $\HS_\LinearPast$ is as expressive as finitary $\CTLStar$.
\end{theorem}

\section{Expressiveness comparison of  $\HS_\LinearTime$,  $\HS_\stat$, and $\HS_\LinearPast$}\label{sect:allSems}

In this section, we compare the expressiveness of the three semantic variants of $\HS$, namely, $\HS_\LinearTime$,  $\HS_\stat$, and $\HS_\LinearPast$. 
The resulting picture was anticipated in Figure \ref{results}. Here, we give the proofs of the depicted results.

We start showing that $\HS_{\stat}$ is \emph{not} subsumed by $\HS_{\LinearPast}$. As a matter of fact, we show that   $\HS_{\stat}$ is sensitive to backward unwinding of finite Kripke structures, allowing us to sometimes discriminate finite Kripke structures with the same computation tree (these structures are always indistinguishable by $\HS_{\LinearPast}$). 

Let us consider, for instance, the two finite Kripke structures $\mathpzc{K}_1$ and $\mathpzc{K}_2$ of Figure~\ref{FigStateBasedExpressiveness}, whose forward and backward unwinding is shown in Figure~\ref{FigFBUnwinding}. Since $\mathpzc{K}_1$ and $\mathpzc{K}_2$ have the same computation tree, no $\HS$ formula $\varphi$ under the computation-tree-based semantics can distinguish $\mathpzc{K}_1$ and $\mathpzc{K}_2$, that is, $\mathpzc{K}_1\models_\LinearPast\varphi$ if and only if $\mathpzc{K}_2\models_\LinearPast\varphi$. On the other hand, the requirement ``\emph{each state reachable from the initial one where $p$ holds has a predecessor where $p$ holds as well}'' can be expressed, under the state-based semantics, by the $\HS$ formula
$\psi:= \hsE(p\wedge \Length_1) \rightarrow \hsE(\Length_1 \wedge \hsAt(p\wedge \neg\Length_1)).$
It is easy to see that $\mathpzc{K}_1\models_\stat\psi$: for any initial trace $\rho$ of $\mathpzc{K}_1$,
we have $\mathpzc{K}_1,\rho\models_\stat \hsE(p\wedge \Length_1)$ iff $\rho=s_0s_1^k$ for $k\geq 1$; the length-1 suffix $s_1$ is \emph{met-by} $s_1s_1$, and $\mathpzc{K}_1,s_1s_1\models_\stat p\wedge \neg\Length_1$.

On the contrary, in $\mathpzc{K}_2$ there is an initial trace, $s_0's_1'$, for which $\mathpzc{K}_2,s_0's_1'\models_\stat\hsE(p\wedge \Length_1)$; however the only traces that meet the length-1 suffix $s_1'$ are $s_1'$ itself and $s_0's_1'$, but neither of them model $p\wedge \neg\Length_1$.
Therefore, $\mathpzc{K}_2\not\models_\stat\psi$. 
This allows us to prove the following proposition.

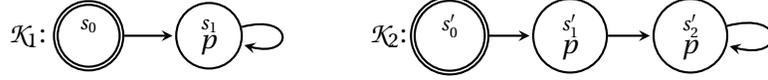
\begin{figure}[t]
\centering
\begin{tikzpicture}[->,>=stealth,thick,shorten >=1pt,node distance=1.6cm]
   \node [right] at (-1.2,0) {$\mathpzc{K}_1$:};
   \node [style={double,circle,draw}] (v0) {$\stackrel{s_0}{\phantom{p}}$};
   \node [style={circle,draw}, right of=v0] (v1) {$\stackrel{s_1}{p}$};
   \node [right] at (3.6,0) {$\mathpzc{K}_2$:};
   \node [right of=v1] (z1) {$\phantom{p}$};
   \node [left of=v0] (z0) {$\phantom{p}$};
   \node [style={double,circle,draw}, right of=z1] (v2) {$\stackrel{s_0'}{\phantom{p}}$};
      \node [style={circle,draw}, right of=v2] (v3) {$\stackrel{s_1'}{p}$};
       \node [style={circle,draw}, right of=v3] (v4) {$\stackrel{s_2'}{p}$};
   \draw (v0) to [right] (v1);
    \draw (v1) to [loop right] (v1);
    \draw (v4) to [loop right] (v4);
    \draw (v2) to [right] (v3);
    \draw (v3) to [right] (v4);
\end{tikzpicture}
\caption{The Kripke structures $\mathpzc{K}_1$ and $\mathpzc{K}_2$.}\label{FigStateBasedExpressiveness}
\end{figure}

\begin{figure}[b]
    \centering
\begin{tikzpicture}[->,>=stealth,thick,shorten >=1pt,node distance=1.6cm]

   \node [right] at (-1.2,0) {$\mathpzc{K}_1$:};
   \node [] (v0) {$\cdots$};
   \node [style={circle,draw}, right of=v0] (v1) {$\stackrel{s_1}{p}$};
   \node [style={double,circle,draw}, above left of=v1] (v1a) {$\stackrel{s_0}{\phantom{p}}$};
   \node [style={circle,draw}, right of=v1] (v11) {$\stackrel{s_1}{p}$};
   \node [style={double,circle,draw}, above left of=v11] (v110) {$\stackrel{s_0}{\phantom{p}}$};
   \node [ right of=v11] (v12) {$\cdots$};
   \draw (v0) to [right] (v1);
    \draw (v1) to (v11);
    \draw (v11) to (v12);
    \draw (v1a) to (v1);
    \draw (v110) to (v11);
\end{tikzpicture}

\bigskip

\begin{tikzpicture}[->,>=stealth,thick,shorten >=1pt,node distance=1.6cm]
   \node [right] at (-1.2,-1) {$\mathpzc{K}_2$:};
   \node [style={double,circle,draw}] (v2) {$\stackrel{s_0'}{\phantom{p}}$};
    \node [style={circle,draw}, right of=v2] (v3) {$\stackrel{s_1'}{p}$};
    \node [style={circle,draw}, below right of=v3] (v4) {$\stackrel{s_2'}{p}$};
    \node [left of=v4] (v4x) {$\cdots$};
    \node [style={circle,draw}, right=2.5cm of v4] (v40) {$\stackrel{s_2'}{p}$};
    \node [style={circle,draw}, above left of=v40] (v401) {$\stackrel{s_1'}{p}$};
    \node [style={double,circle,draw},left of=v401] (v401x) {$\stackrel{s_0'}{\phantom{p}}$};
    \node [right of=v40] (v41) {$\cdots$};
    \draw (v2) to (v3);
    \draw (v3) to (v4);
    \draw (v3) to (v4);
    \draw (v4) to (v40);
    \draw (v40) to (v41);
    \draw (v4x) to (v4);
    \draw (v401) to (v40);
    \draw (v401x) to (v401);
\end{tikzpicture}
\caption{Forward and backward unwinding of $\mathpzc{K}_1$ and $\mathpzc{K}_2$.}\label{FigFBUnwinding}
\end{figure}
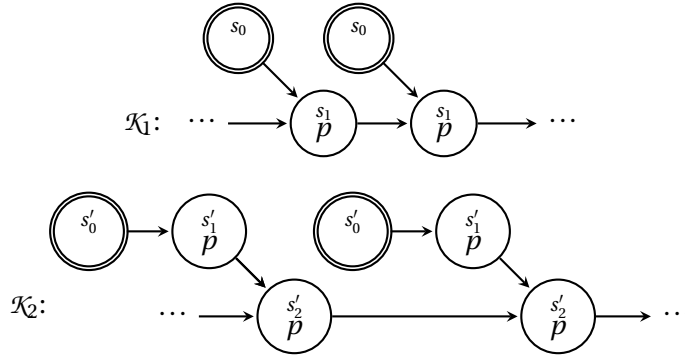

\begin{proposition}\label{Prop:StateBasedExpressiveness} $\HS_{\LinearPast} \not \geq \HS_\stat$.
\end{proposition}

Since, as stated by  Theorem~\ref{Cor:CharacterizationHSCompTree}, $\HS_{\LinearPast}$ and finitary $\CTLStar$ have the same expressiveness and finitary $\CTLStar$ is subsumed by $\HS_\stat$ (Corollary~\ref{cor:FromFInitaryCTLStarToFutureHSBranching}), by Proposition~\ref{Prop:StateBasedExpressiveness} the next corollary follows.

\begin{corollary} $\HS_\stat$ is more expressive than $\HS_{\LinearPast}$.
\end{corollary}

In the following, we focus on the comparison of $\HS_{\LinearTime}$ with $\HS_\stat$ and $\HS_{\LinearPast}$ showing that $\HS_{\LinearTime}$ is incomparable with both $\HS_\stat$ and $\HS_{\LinearPast}$. 

The fact that  $\HS_{\LinearTime}$ does not subsume either $\HS_\stat$ or $\HS_{\LinearPast}$ can be easily proved as follows. Consider the $\CTL$ formula $\forall \Always \exists \Eventually p$ asserting that from each state reachable from the initial one, it is possible to reach a state where $p$ holds. It is well-known that this formula is not $\LTL$-definable (see~\cite{Baier:2008:PMC:1373322}, Theorem~6.21). Thus, by Corollary~\ref{cor:HSLinearCharacterization}, there is no equivalent $\HS_{\LinearTime}$ formula. On the other hand, the requirement $\forall \Always \exists \Eventually p$ can be trivially expressed under the state-based (resp., computation-tree-based) semantics by the $\HS$ formula $\hsBt \hsE p $, proving the following result.

\begin{proposition}\label{prop:nonBranchingExpressibilityOfLinearTime} $\HS_{\LinearTime} \not \geq \HS_\stat$ and $\HS_{\LinearTime} \not \geq \HS_\LinearPast$.
\end{proposition}

To prove the converse, namely, that $\HS_{\LinearTime}$ is not subsumed either by $\HS_\stat$ or by $\HS_{\LinearPast}$, 
we will show that the $\LTL$ formula $\Eventually\, p$ (equivalent to the $\CTL$ formula $\forall \Eventually\,p$)
cannot be expressed in either $\HS_{\LinearPast}$ or $\HS_{\stat}$. The proof is rather involved and requires a number of definitions and intermediate results. We work it out
for the state-based semantics only, because the one for the computation-tree-based semantics is very similar. 



Let us start by defining two families of Kripke structures $(\mathpzc{K}_n)_{n\geq 1}$ and $(\mathpzc{M}_n)_{n\geq 1}$ over $\{p\}$ such that for all $n\geq 1$, the $\LTL$ formula $\Eventually\, p$
 distinguishes $\mathpzc{K}_n$ and  $\mathpzc{M}_n$, and for every $\HS$ formula $\psi$ of size at most $n$, $\psi$ does \emph{not} distinguish
 $\mathpzc{K}_n$ and  $\mathpzc{M}_n$ under the state-based semantics. 
 
For a given $n \geq 1$, the Kripke structures $\mathpzc{K}_n$ and $\mathpzc{M}_n$ are depicted in
 Figure~\ref{FigEventuallyNOnBranchingExpressible}. Notice that the Kripke structure $\mathpzc{M}_n$  differs from  $\mathpzc{K}_n$ only in that its initial state is $s_1$ instead of $s_0$. Formally, $\mathpzc{K}_n=(\{p\},S_n, \delta_n,\mu_n,s_0)$ and $\mathpzc{M}_n=(\{p\},S_n, \delta_n,\mu_n,s_1)$, with
 $S_n=\{s_0,s_1,\ldots, s_{2n},t\}$, $\delta_n=\{(s_0,s_0),(s_0,s_1),\ldots, (s_{2n-1},s_{2n}),\allowbreak (s_{2n},t),(t,t)\}$,  $\mu(s_i)=\emptyset$ for all $0\leq i\leq 2n$, and $\mu(t)=\{p\}$. 
 
 Now, it is immediate to see that
%
 $\mathpzc{K}_n\not\models \Eventually p$ and $\mathpzc{M}_n\models \Eventually p$.

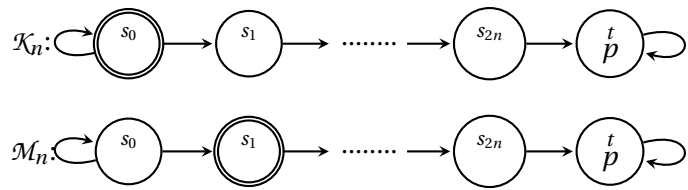
\begin{figure}[b]
\centering
\begin{tikzpicture}[->,>=stealth,thick,shorten >=1pt,node distance=1.6cm]
   \node [right] at (-1.7,0) {$\mathpzc{K}_n$:};
   \node [style={double,circle,draw}] (v0) {$\stackrel{s_0}{\phantom{p}}$};
   \node [style={circle,draw}, right of=v0] (v1) {$\stackrel{s_1}{\phantom{p}}$};
   \draw (v0) to [right] (v1);
       \draw (v0) to [loop left] (v0);
   \node [right of=v1] (v2) {$........$};
   \draw (v1) to [right] (v2);
   \node [style={circle,draw},right of=v2] (v2n) {$\stackrel{s_{2n}}{\phantom{p}}$};
   \draw (v2) to [right] (v2n);
   \node [style={circle,draw}, right of=v2n] (v) {$\stackrel{t}{p}$};
      \draw (v) to [loop right] (v);
   \draw (v2n) to [right] (v);
\end{tikzpicture}

\bigskip

\begin{tikzpicture}[->,>=stealth,thick,shorten >=1pt,node distance=1.6cm]
   \node [right] at (-1.7,0) {$\mathpzc{M}_n$:};
   \node [style={circle,draw}] (v0) {$\stackrel{s_0}{\phantom{p}}$};
   \node [style={double,circle,draw}, right of=v0] (v1) {$\stackrel{s_1}{\phantom{p}}$};
   \draw (v0) to [right] (v1);
       \draw (v0) to [loop left] (v0);
   \node [right of=v1] (v2) {$........$};
   \draw (v1) to [right] (v2);
   \node [style={circle,draw},right of=v2] (v2n) {$\stackrel{s_{2n}}{\phantom{p}}$};
   \draw (v2) to [right] (v2n);
   \node [style={circle,draw}, right of=v2n] (v) {$\stackrel{t}{p}$};
      \draw (v) to [loop right] (v);
   \draw (v2n) to [right] (v);
\end{tikzpicture}
\caption{The Kripke structures $\mathpzc{K}_n$ and $\mathpzc{M}_n$ with $n\geq 1$.}\label{FigEventuallyNOnBranchingExpressible}
\end{figure}

On the contrary, we are going to prove that  $\mathpzc{K}_n\models_\stat \psi$ if and only if $\mathpzc{M}_n\models_\stat \psi$ for all balanced $\HS_\stat$ formulas $\psi$ of length at most $n$ with $n\geq 1$. 
An $\HS_\stat$ formula $\psi$ is \emph{balanced} if, for each subformula $\hsB \theta $ (resp., $\hsBt \theta $), $\theta$ has the form $\theta_1\wedge\theta_2$ with $|\theta_1| = |\theta_2|$. Proving the result for  balanced $\HS_\stat$ formulas allows us to state it for any $\HS_\stat$ formula, since it is possible to trivially 
convert an $\HS_\stat$ formula $\psi$ into a balanced one (by using conjunctions of $\top$) which is equivalent to $\psi$ under any of the considered $\HS$ semantic variants.

To prove such a result, we need some technical definitions.
Let $\rho$ be a trace of $\mathpzc{K}_n$ (note that $\mathpzc{K}_n$ and $\mathpzc{M}_n$ feature the same traces).
By construction, $\rho$ has the form $\rho'\cdot \rho''$, where $\rho'$ is a (possibly empty) trace visiting only states where $p$ does not hold, and $\rho''$
is a (possibly empty) trace visiting only the state $t$, where $p$ holds. We say that $\rho'$ (resp., $\rho''$) is the $\emptyset$-part (resp., $p$-part) of $\rho$.
Let $N_\emptyset(\rho)$, $N_p(\rho)$, and $D_p(\rho)$ be the natural numbers defined as follows:
\begin{itemize}
  \item $N_\emptyset(\rho) = |\rho'|$ (the length of the $\emptyset$-part of $\rho$);
  \item $N_p(\rho) = |\rho''|$ (the length of the $p$-part of $\rho$);
  \item $D_p(\rho)=0$ if $N_p(\rho)>0$ (i.e., $\lst(\rho)=t$); otherwise, $D_p(\rho)$ is the length of the minimal trace starting from  $\lst(\rho)$ and leading to
  $s_{2n}$. Note that $D_p(\rho)$ is well defined and $0\!\leq\! D_p(\rho)\!\leq\! 2n+1$.
\end{itemize}

By construction, the following property holds.

\begin{proposition}\label{remark:HcompatibilityOne} For all traces $\rho$ and $\rho'$ of $\mathpzc{K}_n$, if $D_p(\rho)= D_p(\rho')$, then $\lst(\rho)=\lst(\rho')$.
\end{proposition}

Now, for each $h\in [1,n]$, we introduce the notion of \emph{$h$-compatibility} between traces of  $\mathpzc{K}_n$. Intuitively, this notion provides a sufficient condition to make two traces indistinguishable under the state-based semantics by means of balanced $\HS$ formulas having size at most $h$.

\begin{definition}[$h$-compatibility] Let $h\in [1,n]$. Two traces $\rho$ and $\rho'$ of $\mathpzc{K}_n$ are \emph{$h$-compatible} if the following conditions hold:
\begin{itemize}
\item $N_p(\rho) = N_p(\rho')$;
  \item either $N_\emptyset(\rho) = N_\emptyset(\rho')$, or $N_\emptyset(\rho)\geq h$ and $N_\emptyset(\rho')\geq h$;
        \item either $D_p(\rho) = D_p(\rho')$, or $D_p(\rho)\geq h$ and $D_p(\rho')\geq h$.
\end{itemize}
We denote by $R(h)$ the binary relation over the set of traces of $\mathpzc{K}_n$ such that $(\rho,\rho')\in R(h)$ if and only if $\rho$ and $\rho'$ are $h$-compatible. Notice that $R(h)$ is an equivalence relation, for all $h\in [1,n]$.
Moreover, $R(h)\subseteq R(h-1)$, for all $h\in [2,n]$, that is, $R(h)$ is a refinement of $R(h-1)$.
\end{definition}

By construction, the next property, that will be used to prove Lemma~\ref{lemma:MainnonBranchingExpressibilityOfEventually}, can be easily shown.

\begin{proposition}\label{remark:HcompatibilityTwo} For every trace  $\rho$ of  $\mathpzc{K}_n$ starting from $s_0$ (resp., $s_1$), there exists a trace $\rho'$ of
$\mathpzc{K}_n$ starting from $s_1$ (resp., $s_0$) such that $(\rho,\rho')\in R(n)$.
\end{proposition}

The following lemma lists some useful properties of the equivalence relation $R(h)$.

 \begin{lemma}\label{lemma:Hcompatibility} Let $h\in [2,n]$ and $(\rho,\rho')\in R(h)$. The following properties hold:
 \begin{enumerate}
  \item for each proper prefix $\sigma$ of $\rho$, there exists a proper prefix $\sigma'$ of $\rho'$ such that $(\sigma,\sigma')\in R(\lfloor \frac{h}{2}\rfloor)$;
   \item for each trace of the  form $\rho\cdot \sigma$, where $\sigma$ is not empty, there exists a trace of the form $\rho'\cdot\sigma'$
    such that $\sigma'$ is not empty and $(\rho\cdot\sigma,\rho'\cdot \sigma') \in R(\lfloor \frac{h}{2}\rfloor)$;
 \item for each proper suffix $\sigma$ of $\rho$, there exists a proper suffix $\sigma'$ of $\rho'$ such that $(\sigma,\sigma')\in R(h-1)$;
   \item for each trace of the form $\sigma \cdot \rho$, where $\sigma$ is not empty, there exists a trace of the form $\sigma'\cdot \rho'$ such that $\sigma'$ is not empty and $(\sigma\cdot \rho,\sigma'\cdot \rho')\in R(h)$.
\end{enumerate}
 \end{lemma}
 \begin{proof} We prove Properties 1 and 2. Properties 3 and 4 easily follow by construction and by definition of $h$-compatibility.

 \emph{Property 1.} We distinguish the following cases:
 \begin{enumerate}
  \item $D_p(\rho)<h$ and $N_\emptyset(\rho)<h$. Since $(\rho,\rho')\in R(h)$ and $h\in [2,n]$, it holds that $D_p(\rho)=D_p(\rho')$, $N_\emptyset(\rho)=N_\emptyset(\rho')$, and $N_p(\rho)=N_p(\rho')$, and thus $\rho=\rho'$.
  \item $D_p(\rho)\geq h$. Since $(\rho,\rho')\in R(h)$, $D_p(\rho')\geq h$, $N_p(\rho')=N_p(\rho)=0$,  and either $N_\emptyset(\rho')=N_\emptyset(\rho)$, or $N_\emptyset(\rho)\geq h$ and $N_\emptyset(\rho')\geq h$. In both cases, by construction it easily follows that for each proper prefix $\sigma$ of $\rho$,
      there exists a proper prefix $\sigma'$ of $\rho'$ such that $(\sigma,\sigma')\in R(h-1)\subseteq R(\lfloor \frac{h}{2}\rfloor)$.
  \item $D_p(\rho)<h$ and $N_\emptyset(\rho)\geq h$. Since $(\rho,\rho')\in R(h)$, we have that $D_p(\rho')=D_p(\rho)$ (and hence, by Proposition~\ref{remark:HcompatibilityOne}, $\lst(\rho)=\lst(\rho')$), $N_p(\rho')=N_p(\rho)$,  and $N_\emptyset(\rho')\geq h$.
  
  Let $\sigma$ be a proper prefix of $\rho$. We distinguish the following three subcases:
 \begin{enumerate}
   \item $N_\emptyset(\sigma)< \lfloor \frac{h}{2}\rfloor$. Since $N_\emptyset(\rho)\geq h$, we have that $D_p(\sigma)\geq \lfloor \frac{h}{2}\rfloor$ and $|\sigma|=N_\emptyset(\sigma)$ (and thus $N_p(\sigma)=0$). Since
    $N_\emptyset(\rho')\geq h$, by taking the proper prefix $\sigma'$ of $\rho'$ having length $N_\emptyset(\sigma)$, we obtain that
    $(\sigma,\sigma')\in  R(\lfloor \frac{h}{2}\rfloor)$.
   \item $N_\emptyset(\sigma)\geq  \lfloor \frac{h}{2}\rfloor$ and $D_p(\sigma)\geq  \lfloor \frac{h}{2}\rfloor$. By taking the prefix $\sigma'$ of
   $\rho'$ of length $\lfloor \frac{h}{2}\rfloor$, we get that $(\sigma,\sigma')\in  R(\lfloor \frac{h}{2}\rfloor)$.
   \item  $N_\emptyset(\sigma)\geq  \lfloor \frac{h}{2}\rfloor$ and $D_p(\sigma)<  \lfloor \frac{h}{2}\rfloor$. Since $\lst(\rho)=\lst(\rho')$, $N_p(\rho')=N_p(\rho)$,
   and $N_\emptyset(\rho')\geq h$,
   there exists a proper prefix $\sigma'$ of $\rho'$ such that $\lst(\sigma')=\lst(\sigma)$,  $N_p(\sigma')=N_p(\sigma)$, and
    $N_\emptyset(\sigma')\geq  \lfloor \frac{h}{2}\rfloor$. Hence $(\sigma,\sigma')\in  R(\lfloor \frac{h}{2}\rfloor)$.
 \end{enumerate}
\end{enumerate}
Thus, in all the cases Property~1 holds.

\emph{Property 2.} Let $(\rho,\rho')\in R(h)$ and $\sigma$ be a non-empty trace such that $\rho\cdot \sigma$ is a trace. We distinguish the following cases:
 \begin{enumerate}
  \item $D_p(\rho)<h$. Since $(\rho,\rho')\in R(h)$, we have that $D_p(\rho')= D_p(\rho)$, $N_p(\rho)=N_p(\rho')$, and either $N_\emptyset(\rho')=N_\emptyset(\rho)$, or $N_\emptyset(\rho)\geq h$ and $N_\emptyset(\rho')\geq h$. Hence, $\lst(\rho)=\lst(\rho')$ and, by taking $\sigma'=\sigma$, we obtain that
        $(\rho\cdot\sigma,\rho'\cdot\sigma')\in R(h)\subseteq R(\lfloor \frac{h}{2}\rfloor)$.
  \item $D_p(\rho)\geq h$ and $D_p(\sigma)<\lfloor \frac{h}{2}\rfloor$. It follows that $N_\emptyset(\rho\cdot\sigma)\geq \lfloor \frac{h}{2}\rfloor$. Since
$D_p(\rho')\geq h$, there exists a trace of the form $\rho'\cdot \sigma'$ such that $D_p(\rho'\cdot\sigma')=D_p(\rho\cdot\sigma)$, $N_p(\rho'\cdot\sigma')=N_p(\rho\cdot\sigma)$,
and $N_\emptyset(\rho'\cdot\sigma')\geq \lfloor \frac{h}{2}\rfloor$. Hence
$(\rho\cdot\sigma,\rho'\cdot\sigma')\in R(\lfloor \frac{h}{2}\rfloor)$.
\item $D_p(\rho)\geq h$ and $D_p(\sigma)\geq \lfloor \frac{h}{2}\rfloor$. Thus $D_p(\rho')\geq h$. If $N_\emptyset(\rho\cdot\sigma)<\lfloor \frac{h}{2}\rfloor$, then 
$N_\emptyset(\rho)=N_\emptyset(\rho')$. Therefore, there exists a trace of the form $\rho'\cdot\sigma'$ such that  $N_\emptyset(\rho'\cdot\sigma')=N_\emptyset(\rho\cdot\sigma)$
and $D_p(\sigma')\geq \lfloor \frac{h}{2}\rfloor$. Otherwise, $N_\emptyset(\rho\cdot\sigma)\geq \lfloor \frac{h}{2}\rfloor$ and
   there exists a trace of the form $\rho'\cdot\sigma'$ such that  $N_\emptyset(\rho'\cdot\sigma')\geq \lfloor \frac{h}{2}\rfloor$
and $D_p(\sigma')= \lfloor \frac{h}{2}\rfloor$. In both cases, $(\rho\cdot\sigma,\rho'\cdot\sigma')\in R(\lfloor \frac{h}{2}\rfloor)$.
\end{enumerate}
Thus, Property 2 holds.
\end{proof}

By exploiting Lemma~\ref{lemma:Hcompatibility}, we can prove the following lemma.
 \begin{lemma}\label{corollary:HcompatibilityOne} Let $n$ be a natural number,  $\psi$ be a balanced $\HS_\stat$ formula, with $|\psi|\leq n$, and $(\rho,\rho')\in R(|\psi|)$. Then, $\mathpzc{K}_n,\rho\models \psi$ if and only if $\mathpzc{K}_n,\rho'\models \psi$.
 \end{lemma}
 \begin{proof} The proof is by induction on $|\psi|$. The cases for the Boolean connectives directly follow from the inductive hypothesis and the fact that $R(h)\subseteq R(k)$, for all $h,k\in [1,n]$ with $h\geq k$. 
 
 As for the other cases, we proceed as follows:
 \begin{itemize}
  \item $\psi=p$. Since $(\rho,\rho')\in R(1)$, that is, either $N_{\emptyset}(\rho) = N_{\emptyset}(\rho') = 0$ or both
  $N_{\emptyset}(\rho) \geq 1$ and $N_{\emptyset}(\rho') \geq 1$, $\rho$ visits a state where $p$ does not hold if and only if $\rho'$ visits a state where $p$ does not hold, which proves the thesis.
  \item $\psi =\hsB \theta$ (resp., $\psi =\hsBt\theta$). Since $\psi$ is balanced, $\theta$ has the form $\theta=\theta_1\wedge\theta_2$, with $|\theta_1|=|\theta_2|$. Hence $|\theta_1|,|\theta_2|\leq \lfloor \frac{|\psi|}{2}\rfloor$. We focus on the case $\psi =\hsB \theta$.
  Since $R(|\psi|)$ is an equivalence relation, by symmetry it suffices to show that $\mathpzc{K}_n,\rho\models \psi$ implies $\mathpzc{K}_n,\rho'\models \psi$. If $\mathpzc{K}_n,\rho\models \psi$, then there exists a proper prefix $\sigma$ of $\rho$ such that $\mathpzc{K}_n,\sigma\models \theta_i$, for $i=1,2$. Since $(\rho,\rho')\in R(|\psi|)$, by property (1) of Lemma~\ref{lemma:Hcompatibility}, there exists a proper prefix $\sigma'$ of $\rho'$ such that  $(\sigma,\sigma')\in R(\lfloor \frac{|\psi|}{2}\rfloor)$. Since $R(\lfloor \frac{|\psi|}{2}\rfloor)\subseteq R(|\theta_i|)$, for $i=1,2$, by the inductive hypothesis we get that $\mathpzc{K}_n,\sigma'\models \theta_i$, for $i=1,2$, thus proving that $\mathpzc{K}_n,\rho'\models \psi$.
      
  The case for $\psi =\hsBt \theta$ can be dealt with similarly by exploiting property (2) of  Lemma~\ref{lemma:Hcompatibility}.
  \item $\psi =\hsE \theta$ (resp., $\psi =\hsEt \theta$). We can proceed as in the previous case by applying property (3) of Lemma~\ref{lemma:Hcompatibility} (resp., property (4) of Lemma~\ref{lemma:Hcompatibility}) and the inductive hypothesis.\qedhere
\end{itemize}
 \end{proof}

\begin{lemma}\label{lemma:MainnonBranchingExpressibilityOfEventually} For all natural numbers $n\geq 1$ and balanced $\HS_\stat$ formulas $\psi$, with $|\psi|\leq n$, $\mathpzc{K}_n\models_\stat \psi$ if and only if $\mathpzc{M}_n\models_\stat \psi$.
\end{lemma}
\begin{proof}
First, let us assume that $\mathpzc{K}_n\not\models_\stat \psi$. Then, there exists an initial trace $\rho$ of $\mathpzc{K}_n$ such that $\mathpzc{K}_n,\rho \not \models_\stat \psi$. By Proposition~\ref{remark:HcompatibilityTwo}, there exists a trace $\rho'$ of $\mathpzc{K}_n$, which is an initial trace for $\mathpzc{M}_n$, such that $(\rho,\rho') \in R(|\psi|)$. By Lemma~\ref{corollary:HcompatibilityOne}, we have 
that $\mathpzc{K}_n,\rho ' \not \models_\stat \psi$. Since for any trace $\sigma$ and any $\HS_\stat$ formula $\varphi$, we have that $\mathpzc{K}_n,\sigma \models_\stat \varphi$ if and only if $\mathpzc{M}_n,\sigma \models_\stat \varphi$ ($\mathpzc{K}_n$ and $\mathpzc{M}_n$ feature exactly the same set of traces with exactly the same labeling; they only differ in the initial state), we can conclude that $\mathpzc{M}_n,\rho ' \not \models_\stat \psi$, and thus $\mathpzc{M}_n\not \models_\stat \psi$.

Let us now assume that $\mathpzc{M}_n\not \models_\stat \psi$. Then, there exists an initial trace $\rho$ of $\mathpzc{M}_n$ such that $\mathpzc{M}_n,\rho \not \models_\stat \psi$. As in the converse direction, we have that $\mathpzc{K}_n,\rho \not \models_\stat \psi$, and, by Proposition~\ref{remark:HcompatibilityTwo}, we can easily find an initial trace $\rho'$ of $\mathpzc{K}_n$ such that $(\rho,\rho') \in R(|\psi|)$. By Lemma~\ref{corollary:HcompatibilityOne}, we can conclude that $\mathpzc{K}_n\not\models_\stat \psi$.
\end{proof}

As an immediate consequence of Lemma~\ref{lemma:MainnonBranchingExpressibilityOfEventually} and of the fact that, for each $n\geq 1$,  $\mathpzc{K}_n\not\models \Eventually p$ and $\mathpzc{M}_n\models \Eventually p$, we get the desired undefinability result.




\begin{proposition}\label{prop:nonBranchingExpressibilityOfEventually} The $\LTL$ formula $\Eventually\, p$ (equivalent to the $\CTL$ formula $\forall \Eventually\,p$) cannot be expressed in either $\HS_{\LinearPast}$ or $\HS_{\stat}$.
\end{proposition}

The next proposition immediately follows from Corollary \ref{cor:HSLinearCharacterization} and Proposition~\ref{prop:nonBranchingExpressibilityOfEventually}.

\begin{proposition}\label{prop:oppositeDirection} $\HS_\stat \not \geq  \HS_{\LinearTime} $ and $\HS_\LinearPast \not \geq \HS_{\LinearTime}$.
\end{proposition}

Putting together Proposition~\ref{prop:nonBranchingExpressibilityOfLinearTime} and \ref{prop:oppositeDirection}, we finally obtain the incomparability result.

\begin{theorem} 
$\HS_\LinearTime$ and $\HS_\stat$  are expressively incomparable, and so are $\HS_\LinearTime$ and $\HS_\LinearPast$.
\end{theorem}

The proved results also allow us to establish  the  expressiveness relations between  $\HS_\stat$, $\HS_\LinearPast$ and the standard branching temporal logics $\CTL$ and $\CTLStar$.


\begin{corollary}
The following expressiveness results hold:
\begin{enumerate}
\item $\HS_\stat$ and $\CTLStar$ are expressively incomparable;
\item $\HS_\stat$ and $\CTL$ are expressively incomparable;
\item $\HS_\LinearPast$ and finitary $\CTLStar$ are  less expressive than $\CTLStar$;
\item  $\HS_\LinearPast$ and $\CTL$ are expressively incomparable.
\end{enumerate}
\end{corollary}
\begin{proof}
(Item 1) By Proposition~\ref{prop:nonBranchingExpressibilityOfEventually} and the fact that  $\CTLStar$ is not sensitive to unwinding. 

(Item 2) Again, by Proposition~\ref{prop:nonBranchingExpressibilityOfEventually} and the fact that  $\CTL$ is not sensitive to unwinding.

(Item 3) By Theorem~\ref{Cor:CharacterizationHSCompTree},  $\HS_\LinearPast$ is subsumed by $\CTLStar$, and $\HS_\LinearPast$ and finitary $\CTLStar$ have the same expressiveness. 
Hence, by Proposition~\ref{prop:nonBranchingExpressibilityOfEventually}, the result follows.

(Item 4) Thanks to Proposition~\ref{prop:nonBranchingExpressibilityOfEventually}, it suffices to show that there exists a
$\HS_\LinearPast$ formula which cannot be expressed in $\CTL$. Let us consider the $\CTLStar$ formula
$\varphi:= \EQ\bigl(((p_1 \until p_2) \vee (q_1 \until q_2)) \,\until\, r\bigr)$
over the set of propositions $\{p_1,p_2,q_1,q_2,r\}$. It is shown in \cite{emerson1986sometimes} that $\varphi$
cannot be expressed in $\CTL$. Clearly, if
we replace the path quantifier $\EQ$ in $\varphi$ with the finitary path quantifier $\EQF$,  we obtain an equivalent formula of finitary $\CTLStar$.
Thus, since $\HS_\LinearPast$ and finitary $\CTLStar$ have the same expressiveness (Theorem~\ref{Cor:CharacterizationHSCompTree}),
the result follows.
\end{proof}





\section{Conclusions and future work}

In the present paper, we compared interval temporal logic model checking with point-based one with respect to its expressiveness (and succinctness). To this end, we took into consideration three semantic variants  of  the interval temporal logic $\HS$, namely, $\HS_\stat$, $\HS_\LinearPast$, and $\HS_\LinearTime$, under the homogeneity assumption. We investigated their expressiveness and we systematically contrasted them with the point-based temporal logics $\LTL$, $\CTL$, finitary $\CTLStar$, and $\CTLStar$. 

The resulting picture is as follows: $\HS_\LinearTime$ and $\HS_\LinearPast$ turn out to be as expressive as $\LTL$ and finitary $\CTLStar$, respectively. Moreover, $\HS_\LinearTime$ is at least exponentially more succinct than $\LTL$. $\HS_\stat$ is expressively incomparable with $\HS_\LinearTime$/$\LTL$, $\CTL$, and $\CTLStar$, but it is strictly more expressive than $\HS_\LinearPast$/finitary $\CTLStar$.
We believe it possible to fill the expressiveness gap between $\HS_\LinearPast$ and $\CTLStar$ by considering abstract interval models, induced by Kripke structures, featuring worlds also for infinite traces/intervals, and extending the semantics of $\HS$ modalities to infinite intervals. Such an extension will be investigated in future research.

It is worth noting that the decidability of the MC problem for (full) $\HS_\LinearPast$ and $\HS_\LinearTime$ immediately follows from the above results as a byproduct. We leave for future work the study of the related complexity issues, which have been systematically investigated only for $\HS_\stat$.

MC for $\HS$ can be extended in various directions. Recently~\cite{lm16}, a more general definition of interval labeling, that is, of the behavior of proposition letters over intervals, has been proposed, which allows one to associate a regular expression over the set of states of the Kripke structure with each proposition letter. An in-depth investigation of MC with regular expressions for $\HS$ and its fragments can be found in~\cite{sefm2017,DBLP:journals/corr/abs-1709-02094}, where, in particular, it is shown that MC for full $\HS_\stat$ with regular expressions is still (nonelementarily) decidable, and all the sub-fragments of $\mathsf{A\overline{A}B\overline{B}}_\stat$ and $\mathsf{A\overline{A}E\overline{E}}_\stat$ become complete for $\PSPACE$. 

%

Another research direction looks for possible replacements of Kripke structures by more expressive system models. On one hand, we are interested in the investigation of the MC problem for $\HS$ over \emph{visibly pushdown systems}, that can encode recursive programs and infinite state systems.
On the other, we are thinking of the possibility of devising and exploiting \emph{inherently interval-based models} in system descriptions. Kripke structures, being based on states, are naturally oriented to the representation of the state-by-state evolution of the systems and to the characterization of their point-based properties. To express and check temporal constraints which are inherently interval-based, such as, for instance, those involving temporal aggregations, a different formalism is needed, which allows one to directly model systems on the basis of their interval behavior/properties,
thus making it possible to define and benefit from a really general interval-based MC. 

\end{document}